\documentclass[11pt,reqno]{amsart}
\usepackage[T1]{fontenc}  
\usepackage[utf8]{inputenc}  
\usepackage[english]{babel} 

\usepackage[a4paper,hmargin=3cm,vmargin=3.5cm,heightrounded]{geometry}  
\usepackage{caption} 
\usepackage{mathrsfs} 

\usepackage{graphicx}  
\usepackage[font=scriptsize,labelfont={scriptsize}]{subfig}
\usepackage{color}  

\usepackage{xcolor}
\usepackage{setspace}
\usepackage{enumerate} 
\usepackage[colorlinks=true,linkcolor=blue,urlcolor=blue]{hyperref}
\usepackage[all]{hypcap}

\newcommand{\eps}{\varepsilon}
\renewcommand{\rho}{\varrho}
\renewcommand{\phi}{\varphi}

\newcommand{\N}{\mathbb{N}}

\newcommand{\R}{\mathbb{R}}

\renewcommand{\S}{\mathbb{S}}

\usepackage{bbm} 
\newcommand{\id}{\mathbbm{1}} 

\newcommand{\cA}{\mathcal{A}}
\newcommand{\cB}{\mathcal{B}}

\newcommand{\cD}{\mathcal{D}}

\newcommand{\cF}{\mathcal{F}}

\newcommand{\cI}{\mathcal{I}}

\newcommand{\cM}{\mathcal{M}}

\newcommand{\cP}{\mathcal{P}}

\newcommand{\cW}{\mathcal{W}}

\newcommand{\fN}{\mathfrak{N}}

\renewcommand{\d}{{\rm d}}
\newcommand{\cpl}{{\rm Cpl}}
\newcommand{\mart}{{\rm Mart}}
\newcommand{\Lip}{{\rm Lip}}

\newcommand{\Max}{{\rm Max}}
\newcommand{\sym}{{\rm sym}}
\newcommand{\spn}{{\rm spn\,}}
\newcommand{\Xspn}{X\text{\rm -\,spn\,}}

\newtheorem{theorem}{Theorem}[section]
\newtheorem{lemma}[theorem]{Lemma}
\newtheorem{proposition}[theorem]{Proposition}
\newtheorem{corollary}[theorem]{Corollary}
\theoremstyle{definition}
\newtheorem{definition}[theorem]{Definition}

\newtheorem{remark}[theorem]{Remark}

\numberwithin{equation}{section}

\newcommand{\X}{{H}} 
\newcommand{\Xprime}{{H'}} 

\begin{document}

	\title[Parametric estimation of risk functionals based on Wasserstein distance]{A parametric approach to the estimation of convex risk functionals based on Wasserstein distance}
	
	\author{Max Nendel}
	\address{Center for Mathematical Economics, Bielefeld University}
	\email{max.nendel@uni-bielefeld.de}
	
	\author{Alessandro Sgarabottolo}
	\address{Center for Mathematical Economics, Bielefeld University}
	\email{alessandro.sgarabottolo@uni-bielefeld.de}

	\thanks{The authors thank Daniel Bartl, Jonas Blessing, Stephan Eckstein, Michael Kupper, and Riccardo Mattivi for valuable comments and discussions related to this work.\ Financial support through the Deutsche Forschungsgemeinschaft (DFG, German Research Foundation) -- SFB 1283/2 2021 -- 317210226 is gratefully acknowledged.}
	
	\date{\today}

	\begin{abstract}
		In this paper, we explore a static setting for the assessment of risk in the context of mathematical finance and actuarial science that takes into account model uncertainty in the distribution of a possibly infinite-dimensional risk factor.\ We study convex risk functionals that incorporate a safety margin with respect to nonparametric uncertainty by penalizing perturbations from a given baseline model using Wasserstein distance.\ We investigate to which extent this form of probabilistic imprecision can be approximated by restricting to a parametric family of models.\ The particular form of the parametrization allows to develop numerical methods based on neural networks, which give both the value of the risk functional and the worst-case perturbation of the reference measure.\ Moreover, we consider additional constraints on the perturbations, namely, mean and martingale constraints.\ We show that, in both cases, under suitable conditions on the loss function, it is still possible to estimate the risk functional by passing to a parametric family of perturbed models, which again allows for numerical approximations via neural networks.\smallskip
		
		\noindent \emph{Key words:} Risk measure, model uncertainty, Wasserstein distance, martingale optimal transport, parametric estimation, neural network, measurable direction of steepest ascent\smallskip
		
		\noindent \emph{AMS 2020 Subject Classification:}\ Primary 62G05; 90C31; Secondary 41A60; 68T07; 91G70
		
	\end{abstract}
	
	\maketitle

	\section{Introduction}
In this article, we study a class of convex risk functionals which arises naturally in the context of mathematical finance and actuarial science, when dealing with expected values for a risk factor whose distribution is not perfectly known. Given a random variable $Y$ on a probability space $(\Omega, \cF, \mathbb P)$ taking values in a separable Hilbert space $\X$ endowed with its Borel $\sigma$-algebra $\mathcal B(\X)$, one is usually interested in expressions of the form
\[ \mathbb E_{\mathbb P}\big[f(Y)\big]=\int_\X f(y)\,\mu(\d y),\]
where $f\colon \X\to \R$ is a, say, continuous loss or payoff function and $\mu=\mathbb P\circ Y^{-1}$ is the distribution of $Y$, i.e., a probability measure on $\mathcal B(\X)$.\

Obtaining precise knowledge of the distribution $\mu$ using estimation procedures is therefore of central importance in applications, and lies at the heart of statistics. In practice, however, one often has to deal with statistical imperfections, e.g., a lack of data or information about the dependence structure between single coordinates of $Y$, leading to a so-called \textit{model calibration error} or \textit{model specification error}, respectively.\ Therefore, in many situations, precise knowledge of the underlying distribution $\mu$ of $Y$ may not be at hand, and only a rough estimate or an expert opinion suggesting a particular form of reference distribution $\mu$ may be available.\ This is a special instance of \textit{model uncertainty} appearing, for example, in the context of catastrophic risk in reinsurance or default risk within large credit portfolios in banking.

In the economic literature, model uncertainty is also referred to as \textit{Knightian uncertainty}, and a standard way to deal with it is to look at worst case losses among a set of plausible probability distributions.\ In our study, we follow this approach, and estimate worst case losses over the set $\cP_p$ of all Borel probability measures on $\X$ with finite moment of order $p\in (1,\infty)$, weighting the different measures via a penalization term depending on the \textit{$p$-Wasserstein distance} from a reference model $\mu$ (which is assumed to have finite moment of order $p$ as well). This leads to an expression of the form
\begin{equation}\label{eq.intro.1}
	\cI f:=\sup_{\nu \in \cP_p} \bigg(\int_\X f(z)\,\nu(\d z) - \phi\big(\cW_p(\mu, \nu)\big)\bigg).
\end{equation}
Here, the penalty function $\phi\colon [0,\infty)\to [0,\infty]$, which is assumed to be nondecreasing with $\phi(0)=0$, reflects a degree of confidence in the reference measure, which could, for example, be related to the availability of data for the estimation of $\mu$. The value $\phi(a)=\infty$ for some $a>0$ corresponds to a rejection of every model $\nu$ with Wasserstein distance $\cW_p(\mu,\nu)\geq a$, and the limit case $\phi = \infty \cdot \id_{(0, \infty)}$ resembles perfect confidence in the measure $\mu$.

Functionals of the form \eqref{eq.intro.1} belong to the class of convex risk measures and, under suitable conditions on the penalty function $\phi$, to the class of coherent risk measures, cf.\ \cite{artzner1999} and \cite{FoellmerSchied}. Moreover, they are widely studied in the context of distributionally robust optimization problems, see, for example, \cite{bartl2020computational,gao2016distributionally,esfahani2018data,pflug2007ambiguity,Wozabal2012frameworkambiguity,zhao2018}, where the authors usually consider an additional optimization procedure, leading to an $\inf$-$\sup$-formulation.

A standard approach to tackle the infinite-dimensional optimization related to \eqref{eq.intro.1} is to look for a suitable dual formulation, for example, by transforming the primal problem into a superhedging problem. For example, in \cite{bartl2020computational}, the authors transform a class of robust optimized certainty equivalents (OCEs) into a one-dimensional optimization that leads to an explicit correction term. In general, however, this approach leads to a nested optimization problem, which can be numerically challenging.

We therefore look at this problem from a similar yet different angle, and aim to identify a parametric version of the functional \eqref{eq.intro.1} together with suitable optimizing directions.\
This idea is merely related to the paradigm of looking for extreme points in Wasserstein balls; a topic that has been explored in detail in the case where the reference measure is an empirical distribution (uniform over the samples) or, more generally, a convex combination of Dirac measures. In \cite{Wozabal2012frameworkambiguity}, it is shown that extreme points of Wasserstein balls centered in a measure supported on at most $n$ points are supported on at most $n + 3$ points.\ The paper \cite{Owhadi2015} refines this result, showing that these extreme distributions are in fact supported on $n + 2$ points. Finally, in \cite{esfahani2018data}, the authors show that, under stronger assumptions on the loss function, the infinite-dimensional optimization problem can be solved via a convex shifting of the support points and that the optimizing distribution is supported on the same number of points as the reference measure. In the case, where $\mu$ is a convex combination of Dirac measures, we get a similar result, but in a different fashion, cf.\ Section \ref{sec.numerics.diracs}.

The key idea of our approach is to look for a parametric version of the functional \eqref{eq.intro.1} in terms of a first order approximation as the level of uncertainty related to the penalty function $\phi$ tends to zero. More precisely, we introduce a scaling parameter $h>0$, and substitute $\phi$ with the rescaled version $\phi_h:=h \phi(\cdot / h)$, which allows to control the level of uncertainty in terms of $h$. We then consider the operator $\cI(h)$, given by
\begin{equation}\label{eq.intro.2}
	\cI(h)f:=\sup_{\nu \in \cP_p} \bigg(\int_\X f(z)\,\nu(\d z) - \phi_h\big(\cW_p(\mu, \nu)\big)\bigg)
\end{equation}
for a suitable class of continuous payoff function $f\colon \X\to \R$, and look for a parametric version $\cI_\Theta(h)$ that asymptotically coincides with $\cI(h)$ up to a first order error as $h$ tends to zero. To that end, we consider a \textit{parameter set} $\Theta$ consisting of vector fields $\theta \colon \X\to \X$, which are $p$-integrable with respect to the reference measure $\mu$. For each \textit{parameter} $\theta\in \Theta$, we consider a probability measure $\mu_\theta$, which is a shifted version of $\mu$, i.e.,
\begin{equation} \label{eq.intro.normal.1}
	\int_\X f(z)\,\mu_\theta(\d z)=\int_\X f\big(y+\theta(y)\big)\,\mu(\d y),
\end{equation}
and define the parametric version $\cI_\Theta(h)$ of $\cI(h)$, for $h>0$, by
\[
\cI_\Theta (h)f:=\sup_{\theta\in\Theta }\bigg(\int_\X f(z)\, \mu_\theta(\d z)-\phi_h\big( \| \theta \|_{L_p(\mu; \X)}\big)\bigg).
\]
In Theorem \ref{thm.main.itheta}, we provide conditions on the parameter set $\Theta$ and the function $f$, ensuring that
\[
\lim_{h\downarrow 0}\frac{\cI(h)f-\cI_\Theta (h)f}{h}=0.
\]
In Theorem \ref{cor.main.itheta.single}, we compute a safety margin for asymptotically small $h>0$, and show that an asymptotically optimal parameter $\theta$ can be found by looking at directions of steepest ascent for $f$.

A crucial step in the direction of studying the asymptotic behaviour of optimization problems of the form \eqref{eq.intro.2} was done in \cite{bartl2021sensitivity}, see also \cite{bartlwiesel2022} for its extension to a multi-period setting using adapted Wasserstein distance and \cite{BartlEckstKuppRandWalks2021,nendel2021wasspert} for dynamic versions without an additional optimization.\ In  \cite{bartl2021sensitivity}, the authors study sensitivities for various forms of robust optimization problems over $p$-Wasserstein balls employing an elegant computational approach.\ In the proof of Theorem \ref{cor.main.itheta.single}, we build on these methods, allowing for more general penalty functions $\phi$ and, at the same time, using weaker differentiability assumptions on $f$.\ In particular, we can drop the assumption of  differentiability ($\mu$-a.e.), introducing the concept of a \textit{measurable direction of steepest ascent} as a generalization of the gradient, cf.\ Definition \ref{def.measdirec}. This way, we can overcome differentiability issues related to the function $f$ for reference measures $\mu$, which are not regular. Most functions of interest in financial applications, e.g., call options, have a measurable direction of steepest ascent, so that our results can be applied without any restrictions on the reference measure $\mu$, cf.\ Remark \ref{rem.independ.measure} b) for a more thorough discussion.

Following the lead of \cite{bartl2021sensitivity}, in Section \ref{sec.main.mean.constr}, we study an additional mean constraint in the optimization \eqref{eq.intro.2}. This constraint enters naturally when dealing with risk-neutral pricing, where the mean of the underlying is assumed to be known (e.g., by standard non-arbitrage arguments).\ Thanks to our parametric description of the risk functional, cf.\ Theorem \ref{thm.mean.constraint}, we can show that uncertainty in the return of a financial position can be replicated by means of self-financing portfolios of call options or digital options, cf.\ Section \ref{sec.example.financial}.

In Section \ref{sec.main.martingale}, we impose a so-called martingale constraint on the functional \eqref{eq.intro.2}. That is, we restrict the optimization to a set of probability measures, which are given in terms of a martingale perturbation of the reference measure or, in different words, measures that admit a martingale coupling with $\mu$.\ This setup is closely intertwined with the topics of martingale optimal transport (MOT) and model-free pricing in mathematical finance.\ For example, in \cite{Beiglbck2013ModelindependentBF}, the authors study robust superhedging problems based on MOT, whereas \cite{Beiglbck2016} provides a precise description and fine properties of the optimal transport plan under a martingale constraint.\ We also refer to \cite{OblojWiesel2021} for a class of robust estimators for superhedging prices based on  martingale measures which are, up to a small perturbation in Wasserstein distance, equivalent to an empirical distribution.

In Theorem \ref{thm.martingale.constraint}, we provide an asymptotic parametrization of the constrained version of \eqref{eq.intro.2} through a suitable \emph{randomization} of the reference measure. Given a direction $\theta\in \Theta$, we include an additional coin flip in  \eqref{eq.intro.normal.1}, which is independent of $\mu$ and determines the sign of $\theta$. This leads to a parametric description in terms of measures $\mu_\theta^\mart\in \cP_p$, given by
\begin{equation} \label{eq.intro.mart.1}
	\int_\X f(z)\,\mu_\theta^\text{Mart}(\d z)=\int_\X \int_\R f\big(y+s\theta(y)\big)\,B_{\sym}(\d s)\,\mu(\d y),
\end{equation}
where $B_{\sym}$ is the symmetric Bernoulli distribution on $\R$ with equal probabilities, i.e., $$B_{\sym}\big(\{-1\}\big)=B_{\sym}\big(\{1\}\big)=\frac{1}2.$$ 
In principle, $B_{\sym}$ could be replaced by any distribution on $\R$ with mean zero. However, the symmetric Bernoulli distribution has the identifying property that, apart from having mean zero, $\int_\R |s|^\alpha\, B_\sym(\d s)=1$ for all $\alpha\in (0,\infty)$; a property that is of central importance for the asymptotic parametrization in this framework.

Apart from this, the symmetric Bernoulli distribution is fundamentally connected to the Brownian motion via Donsker's theorem.\ Having in mind that, in a Brownian filtration, all martingales can be represented as stochastic integrals with respect to the Brownian motion, our parametrization via measures of the form \eqref{eq.intro.mart.1} can be seen as a microscopic version of a martingale representation theorem in a completely different and model-free setting.\ Moreover, the measure $\mu_\theta^\mart$ can be represented via the explicit formula
\begin{equation} \label{eq.intro.mart.2}
	\int_\X f(z)\,\mu_\theta^\text{Mart}(\d z)=\int_\X \frac{f\big(y+\theta(y)\big)+f\big(y-\theta(y)\big)}{2}\,\mu(\d y),
\end{equation}
which, after substracting $\int_\X f(y)\, \mu(\d y)$, leads to a finite difference approximation of the second derivative of $f$ in the direction $\theta$ using central differences.\ This links $\mu_\theta^\mart$ also from an analytic perspective to a Brownian motion through its infinitesimal generator.

In view of numerical approximations, our construction of parametric versions of convex risk functionals of the form \eqref{eq.intro.2} leads to a description of asymptotically relevant models for the optimization in terms of Monge transports ($p$-integrable vector fields), which, in turn, suggests a numerical investigation of the risk functional in the spirit of \cite{ecksteinNN2021}, see also \cite{eckstein2020robustrisk}.\ In Section \ref{sec.numerics.abscont}, we develop a numerical scheme based on neural networks. Previous works on this topic provide approximations from above based on duality results, whereas our approach leads to an approximation from below, based on the restriction to a parametric family of models.
As a byproduct, we also approximate the optimizing measure through the Monge transport generating it. Observing that, for small values of uncertainty, the optimal transport plan depends on the reference measure only through a multiplicative factor, we draw a connection to \textit{transfer learning}, and discuss the robustness of our approach with respect to deviations in the reference measure.

In machine learning, {transfer learning} usually consists of taking a neural network trained on a previous dataset and training only its last layer on a new dataset, which is often significantly smaller than the first one, see \cite{Bozinovski2020} for references to seminal works on this topic.\ The idea is that if the two datasets share some common features, the part of the network that is inherited from the first training is able to extract most of these features, while the training of the last layer learns from specific characteristics of the new dataset.\ In our framework, we can explicitly distinguish between the common feature that can be extracted from the first training, namely the optimizing vector field, and the feature that is specific to each reference measure, i.e., the rescaling factor.\ In particular, once the approximation is performed on one measure, we can transfer it to a different measure at the price of a one-dimensional optimization, see Section \ref{sec.example.transferlearning} for further details.

The paper is organized as follows.\ In Section \ref{sec.main}, we introduce the setup and state the main results on parametric versions of risk functions of the form \eqref{eq.intro.2}. Section \ref{sec.numerics} provides numerical methods for the approximation of the parametric risk functional $\cI_\Theta(h)$.\ We distinguish between two situations leading to different numerical schemes.\ One is based on a reduction to a finite-dimensional optimization, cf.\ Section \ref{sec.numerics.diracs}, while the other uses an approximation via neural networks, cf.\ Section \ref{sec.numerics.abscont}. In Section \ref{sec.examples}, we discuss applications, both, of our theoretical and numerical findings by means of examples from finance and insurance.\ Section \ref{sec.proof.main} contains the proofs of the main theorems and, in the Appendix \ref{app.A}, we provide a simple approximation result that helps to reduce the set of parameters $\Theta$.

\section{Setup and main results} \label{sec.main}
In this section, we introduce the class of convex risk functionals that (together with additional constraints) form the center of our study, and we state our main results concerning their parametric estimation.

Throughout, let $p \in (1, \infty)$, $q:=\frac{p}{p-1}\in (1,\infty)$ be the conjugate exponent of $p$, and $(\X,\langle\,\cdot\,,\, \cdot\,\rangle)$ be a separable Hilbert space.\ As usual, we endow $\X$ with its canonical norm $\|\cdot\|:=\sqrt{\langle\,\cdot\,,\, \cdot\,\rangle}$, and identify the topological dual space $\Xprime$ of $\X$ with $\X$ itself. We denote the set of all probability measures $\nu$ on the Borel $\sigma$-algebra $\mathcal B(\X)$ of $\X$ 
with
$$|\nu|_p:=\left(\int_\X\, \|y\|^p\, \nu(\d y)\right)^{1/p}<\infty$$ by $\cP_p=\cP_p(\X)$. 

In the following, we consider a fixed probability measure $\mu\in \cP_p$, which we will refer to as the \textit{reference measure} or \textit{baseline model}, and a \textit{penalty function} $\phi \colon [0, \infty) \to [0, \infty]$, which is assumed to be nondecreasing with $\phi(0)=0$.

For $\nu\in \cP_p$, we denote the \textit{$p$-Wasserstein distance} between the reference measure $\mu$ and $\nu$ by  \begin{equation}\label{def.wasserstein}
	\cW_p(\mu,\nu) = \left(\inf_{\pi \in \cpl(\mu,\nu)} \int_{\X \times \X} \|y-z\|^p\,\pi(\d y,\d z)\right)^{1/p},
\end{equation}
where $\cpl(\mu, \nu)$ is the set of all couplings between $\mu$ and $\nu$, i.e., the set of probability measures on the Borel $\sigma$-algebra $\mathcal B(\X\times \X)$ of the product space $\X \times \X$ with first and second marginal $\mu$ and $\nu$, respectively. We refer to \cite{ambrosio2008gradient} and \cite{villani2008optimal} for a detailed discussion on Wasserstein distances and, more generally, the topic of optimal transport. Here, we only recall that, for all $\nu\in \cP_p$, there exists an optimal coupling $\pi^*\in \cpl(\mu,\nu)$ that attains the infimum in \eqref{def.wasserstein}, i.e.,
\[
\cW_p(\mu,\nu) = \left( \int_{\X \times \X} \|y-z\|^p\,\pi^*(\d y,\d z)\right)^{1/p}.
\]

For $\rho\in [1,\infty)$, we denote the space of all ($\mu$-equivalence classes of) measurable functions $f\colon \X\to \R$ with $$\|f\|_{L_\rho(\mu)}:=\bigg(\int_\X |f(y)|^\rho\, \mu(\d y)\bigg)^{1/\rho}<\infty$$ by $L_\rho(\mu)$. In a similar fashion, $L_\rho(\mu;\X)$ denotes the space of all ($\mu$-equivalence classes of) measurable functions $g\colon \X\to \X$ with $$\|g\|_{L_\rho(\mu;\X)}:=\bigg(\int_\X \|g(y)\|^\rho\, \mu(\d y)\bigg)^{1/\rho}<\infty.$$ 
Recall that, by assumption, $\X$ is a separable Hilbert space, so that the image $g(\X)\subset \X$ of $g$ is automatically separable.

\subsection{The unconstrained case}

Throughout this subsection,
we assume that
\begin{equation} \label{eq: property of phi}
	\liminf_{v \to \infty} \frac{\phi(v)}{v^{p}} > 0.
\end{equation}
As a consequence, the \textit{convex conjugate} $\phi^*\colon  [0,\infty)\to [0,\infty)$, given by
\begin{equation}\label{eq.def.conjugate}
	\phi^*(u):=\sup_{v\geq 0} \Big(uv-\phi(v)\Big)\quad \text{for all }u\in [0,\infty),
\end{equation}
is well-defined, convex, and continuous with $\phi^*(0)=0$.

We denote by $\Lip_{p}$ the space of all functions $f\colon \X\to \R$, for which there exists a constant $L_f\geq 0$ such that
\begin{equation}\label{eq.lipp}
	|f(x_1)-f(x_2)|\leq L_f\big(1+\max\{\|x_1\|,\|x_2\|\}\big)^{p-1}\|x_1-x_2\|\quad \text{for all }x_1,x_2\in \X.
\end{equation}
Let $f\in \Lip_p$. Choosing $x_1=x\in \X$ and $x_2=0$  in \eqref{eq.lipp}, we find that
\begin{equation}\label{eq.growth.lip}
	|f(x)|\leq (|f(0)|+L_f)\big(1+\|x\|\big)^{p}\quad \text{for all }x\in \X.
\end{equation}
Moreover, the \textit{local Lipschitz constant} $|f|_\Lip\colon \X\to \R$ of $f$, given by
\begin{equation}\label{eq.lipschitzconstant}
	|f|_\Lip(x):=\limsup_{h\downarrow 0}\sup_{\|u\|\in\{0,1\}} \frac{f(x + h u) - f(x)}{h} \quad\text{for all }x\in \X,
\end{equation}
is measurable due to the continuity of $f$, and satisfies
\begin{equation}
	0\leq |f|_\Lip(x)= \limsup_{h\downarrow 0}\sup_{\|u\|\in\{0, 1\}} \frac{f(x + h u) - f(x)}{h}\leq L_f\big(1+\|x\|\big)^{p-1}\quad\text{for all }x\in \X.
\end{equation}

For $h>0$, we consider the operator $\cI(h)$ that associates to every function $f\in \Lip_p$ the quantity
\begin{equation}\label{def.I}
	\cI(h)f:=\sup_{\nu \in \cP_p} \left( \int_{\X} f(z)\,\nu(\d z) - \phi_h (\cW_p(\mu, \nu)) \right)\in (-\infty,\infty],
\end{equation}
where $\phi_h(v) := h \phi\left(\frac{v}{h}\right)$ for all $v\in [0,\infty)$. The fact that $f\in \Lip_p$ ensures that, at least for sufficiently small $h>0$,
\begin{align*}
	\cI(h)f<\infty.
\end{align*}

\begin{remark}
	Considering the extreme case, where $\phi=\infty \cdot\id_{(a, \infty)}$ for some $a\geq 0$, the operator $\cI(h)$ simplifies to
	\[
	\cI(h)f:=\sup_{\cW_p(\mu,\nu)\leq ah} \int_{\X} f(z)\,\nu(\d z)\quad \text{for all }f\in \Lip_p,
	\]
	which corresponds to the consideration of worst case scenarios over Wasserstein balls with radius $ah$ for $h>0$. Thus, even in the simplest case, the operator $\cI(h)$ leads to a rather complex optimization problem, which involves the determination of a Wasserstein ball. Moreover, due to the geometric structure (more precisely, the lack of finitely many extreme points) of Wasserstein balls, it is, in general, not trivial to characterize optimizers of $\cI(h)$ and to analyze their dependence on the loss function $f$ and the reference model $\mu$. One aspect that triggers this problematic is the appearance of the Wasserstein distance in the penalty term. On the other hand, also the lack of concrete structure of the set $\cP_p$ makes the problem intractable. In the following, we thus introduce a simplified version of the operator $\cI(h)$, which depends on the reference measure in a parametric way, circumvents the use of the Wasserstein distance, and can, in many cases, be computed in an  efficient manner.
\end{remark}

Throughout, we consider a \textit{parameter set} $\Theta \subseteq L_p(\mu;\X)$. For every \textit{parameter} $\theta\in \Theta$, we define a probability measure $\mu_\theta\in \cP_p$ via
\[
\int_\X f(z)\,\mu_\theta(\d z):=\int_\X f\big(y+\theta(y)\big)\, \mu(\d y)\quad \text{for all } f \in \Lip_p.
\]
Note that, for all $\theta \in \Theta$, $\mu_\theta$ is in fact an element of $\cP_p$ since
\[
|\mu_\theta|_p \le |\mu|_p + \|\theta\|_{L_p(\mu;\X)}.
\]
For $h>0$, we define the \textit{parametric version} $\cI_\Theta$ of $\cI$ by
\[
\cI_\Theta (h)f:=\sup_{\theta\in\Theta }\bigg(\int_\X f(z)\, \mu_\theta(\d z)-\phi_h\big( \| \theta \|_{L_p(\mu; \X)}\big)\bigg)\quad \text{for all }f\in \Lip_p.
\]
By definition, $\cI_\Theta(h)f\leq \cI(h)f$ for all $h> 0$ and $f\in \Lip_p$. In fact, as already observed, $\{\mu_\theta\}_{\theta \in \Theta} \subset \cP_p$, and to bound the penalization term we choose a perfectly correlated coupling and get that $\cW_p(\mu, \mu_\theta) \le \|\theta\|_{L_p(\mu; \X)}$.
In conclusion, for $h>0$, the operator $\cI_\Theta(h)$ restricts the family of measures appearing in the optimization problem \eqref{def.I} and, at the same time, has a more explicit penalization term.

\begin{remark} \label{rem.abscont.param}
	Assume that $\mu$ is nonatomic.\ Then, by a standard construction, there exists a measurable function $f\colon \X\to [0,1]$ that is uniformly distributed under $\mu$. In fact, since $\mu$ is nonatomic, for every set $B\in \cB(\X)$, there exists a set $A\in \cB(\X)$ with $A\subset B$ and $\mu(A)=\frac12 \mu(B)$, so that there exists an i.i.d.\ sequence $(f_n)_{n\in \N}$ of $\{0,1\}$-valued random variables with $$\mu(f_n=0)=\frac12=\mu(f_n=1)\quad \text{for all }n\in \N.$$ Choosing $f=\sum_{n=1}^\infty2^{-n}f_n$, we find that $f\colon \X\to [0,1]$ is uniformly distributed. On the other hand, since $\X$ is a separable Hilbert space, it is Borel isomorphic to $\R$, cf.\ \cite[Theorem 13.1.1]{MR1932358}. Therefore, every measure $\nu\in \cP_p$ can be obtained as a push forward measure $\mu\circ g^{-1}$ for a measurable function $g\colon \X\to \X$. Choosing, $\theta(x):=g(x)-x$ for all $x\in \X$, it follows that $\nu=\mu_\theta$. Since both $\mu$ and $\nu$ are elements of $\cP_p$, we obtain that $\theta\in L_p(\mu;\X)$. Hence, for every $\nu\in \cP_p$, there exists some $\theta\in L_p(\mu;\X)$ with $\nu=\mu_\theta$. On the other hand, if $\mu$ is even assumed to be regular, there exists some $\theta\in L_p(\mu;\X)$ with $\nu=\mu_\theta$ and
	\begin{equation}\label{eq.optimizerWp}
		\cW_p(\mu,\nu)=\|\theta\|_{L_p(\mu;\X)},
	\end{equation}
	see, e.g., \cite[Theorem 6.2.10]{ambrosio2008gradient}. In conclusion, if $\mu$ is regular, then
	\begin{equation}\label{eq.IequalItheta}
		\cI(h)f=\cI_{L_p(\mu;\X)}(h)f\quad \text{for all }f\in \Lip_p \text{ and all } h>0.
	\end{equation}
	However, one should keep in mind the following two observations:
	\begin{enumerate}[1)]
		\item Since the starting point of our study is a somewhat imperfect knowledge of the reference measure $\mu$, it cannot be assumed to be nonatomic. In particular, if $\mu$ is estimated from data as a, say, empirical distribution, it will be of the form $\mu=\sum_{i=1}^n\alpha_i\delta_{x_i}$ with $n\in \N$, $x_1,\ldots, x_n\in \X$, and $\alpha_1,\ldots,\alpha_n\in [0,1]$ with $\sum_{i=1}^n\alpha_i=1$, where $\delta_x$ denotes the Dirac measure with barycenter $x\in \X$. 
		\item \label{rem.abscont.param.instability} Even if $\mu$ is regular, and we have \eqref{eq.IequalItheta} at hand, it neither reveals a particular shape nor qualitative properties of the optimizer.\ Moreover, a priori, the optimizer may heavily depend on the reference measure $\mu$, which is assumed not to be known precisely.\ Therefore, the computation might exhibit huge instabilities for small deviations from the reference measure.
	\end{enumerate}
	We therefore aim towards a different direction, and try to come up with an asymptotic version of the equality \eqref{eq.IequalItheta} for infinitesimally small $h>0$ and \textit{all} reference measures $\mu$.\ On the other hand, we are interested in understanding the structure of asymptotic optimizers of $\cI(h)$ and their dependence on the reference measure $\mu$ for infinitesimally small levels of uncertainty $h>0$.
\end{remark}

Every element of $\Lip_p$ is locally Lipschitz continuous, and thus Gateaux differentiable outside a Gaussian null set, see \cite[Theorem 5.11.1]{bogachev1998gaussian}. However, since we do not want to restrict our attention to regular measures, i.e., the ones that assign measure zero to every Gaussian null set, we make use of a similar yet slightly different notion of differentiability based on the idea that the gradient is the direction of steepest ascent.\ The following definition formalises this idea.

\begin{definition}\label{def.measdirec}
	Let $f\in \Lip_p$.\ We say that a vector field $v\colon \X\to \X$ is a \textit{measurable direction of steepest ascent} for $f$, if it is measurable, $\|v(x)\|\in \{0,1\}$, and 
	\[
	|f|_\Lip(x)= \lim_{h\downarrow 0} \frac{f\big(x+hv(x)\big)-f(x)}{h}\quad\text{for all }x\in \X.
	\]
\end{definition}

\begin{remark}\label{rem.meas.direc}
	Let $f\colon \X\to \R$ be Fr\'echet differentiable with continuous gradient $\nabla f\colon \X\to \X$ satisfying
	\[
	L_f:=\sup_{x\in \X}\frac{\|\nabla f(x)\|}{(1+\|x\|)^{p-1}}<\infty.
	\]
	Then,
	\begin{align*}
		\big|f(x_1)-f(x_2)|&\leq \int_0^1\big|\big\langle\nabla  f\big(tx_1+(1-t)x_2\big),x_1-x_2\big\rangle\big|\, \d t\\
		&\leq L_f\big(1+\max\{\|x_1\|,\|x_2\|\}\big)^{p-1}\|x_1-x_2\|\quad\text{for all }x_1,x_2\in \X,
	\end{align*}
	so that $f\in \Lip_p$.\ In this case, a measurable direction of steepest ascent $v\colon\X\to \X$ is given by 
	\begin{equation} \label{eq.directioncont}
		v(x):= \begin{cases}
			\frac{\nabla f(x)}{\|\nabla f(x)\|}, & \text{if } \nabla f(x) \ne 0, \\
			v(x) = 0, & \text{otherwise,}
		\end{cases}\quad\text{for all }x\in \X,
	\end{equation}
	and $|f|_\Lip (x) = \|\nabla f(x)\|$ for all $x \in \X$.
\end{remark} 
The following theorem provides a first order approximation of the functional $\cI(h)$ in terms of its parametric version $\cI_\Theta(h)$, for every reference measure $\mu$ and infinitesimally small $h>0$, under weak regularity assumptions on the function $f$.\ Additionally, it shows that an asymptotic optimizer can be computed explicitly and (almost) independently of the reference measure $\mu$. The proof is delayed to Section \ref{sec.proof.main}.

\begin{theorem}\label{cor.main.itheta.single}
	Let $f\in \Lip_p$ with a measurable direction of steepest ascent $v\colon \X\to \X$. Define
	\begin{equation}\label{eq.opt.theta}
		\theta(x) = v(x) \big(| f |_\Lip(x)\big)^{q-1} \quad \text{for all } x \in \X.
	\end{equation}
	Then,
	\[
	\lim_{h \downarrow 0} \frac{\cI(h)f - \cI_\Theta(h)f}{h} = 0,
	\]
	whenever $\{\lambda\theta\, |\,\lambda\geq 0\} \subseteq \Theta$. Moreover,
	\begin{equation}\label{eq.main.itheta}
		\lim_{h \downarrow 0} \frac{\cI(h)f- \mu f}{h} = \phi^* \Big( \big\||f|_\Lip\|_{L_q(\mu)} \Big),
	\end{equation}
	where $\mu f:=\int_\X f(y)\, \mu(\d y)$.
\end{theorem}

\begin{remark} \label{rem.independ.measure}
	\begin{enumerate}[a)]
		\item The previous result shows us that, when considering a fixed function in $\Lip_p$ with a measurable direction of steepest ascent, the optimization problem \eqref{def.I} asymptotically reduces to a one-dimensional optimization scheme as $h\downarrow 0$. Moreover, this scheme is independent of the reference measure $\mu$. We exploit this property in Section \ref{sec.example.transferlearning} to develop a \emph{transfer learning} technique, where we take the optimizing vector field of a given problem and search for optimizers with respect to different reference measures by simply rescaling the original vector field. This is quite remarkable in view of the concerns expressed in point \ref{rem.abscont.param.instability}) of Remark \ref{rem.abscont.param}, since it indicates sort of a stability of the optimization problem $\cI_\Theta(h)$ with respect to perturbations of the reference measure, at least for small $h>0$.\ We further mention that the vector field $\theta$, given by \eqref{eq.opt.theta}, satisfies
		\[
		\sup_{x\in \X}\frac{\|\theta(x)\|}{1+\|x\|}<\infty.
		\]
		\item The conditions of the previous theorem might, at first sight, seem a bit artificial, and, in view of Remark \ref{rem.meas.direc}, a more natural formulation seems to be in terms of continuously differentiable functions $f\colon \X\to \R$. However, for instance the payoff function of a European call option with \textit{strike price} $K\in \R$ ($f(x) =
		(x - K)^+$) is not differentiable, since its derivative has a discontinuity at the strike price $K$. In \cite{bartl2021sensitivity}, when computing the sensitivity of the robust optimization problem, the authors remark that the assumption of differentiability can be quite naturally relaxed when the measure $\mu$ is absolutely continuous with respect to the Lebesgue measure, assuming that the function $f$ admits a weak derivative in the Sobolev sense, which is continuous $\mu$-a.e.\ They also show how one can avoid the absolute continuity of $\mu$ with respect to the Lebesgue measure using a regularization procedure via convolution with normal distributions. Nevertheless, this solution does not solve the issue arising from $f$ having a kink on a set which has strictly positive measure under $\mu$. This is the case for a call option with strike $K$ and a reference measure $\mu$ with $\mu(\{K\})>0$.
		In our framework, we can deal with this particular case observing that, if $f$ is the profile of the call option with strike $K$,
		\[
		|f|_\Lip(x) = \begin{cases}
			0, & \text{for } x < K, \\
			1, & \text{for } x \ge K, \end{cases}
		\]
		and a measurable direction of steepest ascent is given by $v(x) = \id_{[K,\infty)}(x)$ for all $x\in \R$.
	\end{enumerate}
\end{remark}

If $\Theta$ is dense in $L_p(\mu;\X)$ with $0\in \Theta$, then
\begin{equation}\label{eq.equalitylp}
	\cI_\Theta(h)f=\cI_{L_p(\mu;\X)}(h)f\quad\text{for all }f\in \Lip_p.
\end{equation}
A combination of Theorem \ref{cor.main.itheta.single} and this insight, cf.\ Lemma \ref{lem.approx}, is the essence of the proof of the following theorem, which is our second main result of this subsection.
\begin{theorem} \label{thm.main.itheta}
	Assume that $\Theta$ is dense in $L_p(\mu;\X)$ with $0\in \Theta$. Then,
	\[
	\lim_{h \downarrow 0} \frac{\cI(h)f - \cI_\Theta(h)f}{h} = 0
	\]
	for every $f \in \Lip_p$ that has a measurable direction of steepest ascent.
\end{theorem}

In the Appendix \ref{app.A}, we provide a simple sufficient condition for the density of $\Theta$ in $L_p(\mu;\X)$, which allows to perform numerical computations of the operator $\cI(h)$ in Section \ref{sec.numerics} via sufficiently small sets $\Theta$ and to make use of a neural network approximation of the operator $\cI(h)$, for $h>0$. In Section \ref{sec.examples}, we also discuss approximations, which are based on the use of portfolios of European call options or digital options.

\subsection{Mean constraint} \label{sec.main.mean.constr}

In this section, we consider a constrained version of the operator $\cI(h)$, for $h>0$, where the supremum is taken over all measures $\nu \in \cP_p$ with $$\int_\X z\,\nu(\d z) = \int_\X y\,\mu(\d y),$$ i.e., over all measures that have the same mean as the reference measure $\mu$.\ This kind of constraint arises naturally in several contexts, for example, when pricing a financial instrument under the assumption that the expected value of the underlying is known.

Again, we assume that the penalty function satisfies \eqref{eq: property of phi}. In the following, we restrict the operator $\cI$ to measures, which satisfy the mean constraint, and consider, for $h>0$ and $f\in \Lip_p$,
\[
\cI^{\rm Mean}(h) f := \sup_{\nu \in \cP_p(\mu)} \bigg( \int_\X f(z)\,\nu(\d z) - \phi_h \big({\cW_p(\mu, \nu)}\big) \bigg),
\]
where $\cP_p(\mu)$ denotes the set of all $\nu\in \cP_p$ with  $\int_\X z\,\nu(\d z) = \int_\X y\,\mu(\d y)$.

We call $C^1_p$ the space of all continuously (Fr\'echet) differentiable functions $f \colon \X \to \R$, for which there exists a constant $C\geq 0$ such that
\[
\lVert \nabla f(x) \rVert \le C(1 + \|x\|)^{p-1}\quad \text{for all }x\in \X.
\]

It is immediate to see that, for $\theta \in L^p(\mu;\X)$, the measure $\mu_\theta$ is an element of $\cP_p(\mu)$ if and only if $\int_\X\theta(y)\,\mu(\d y)=0$.

\begin{theorem} \label{thm.mean.constraint}
	Let $p = 2$ and
	$f \in C^1_p$. Then,
	\[
	\lim_{h \downarrow 0} \frac{\cI^{\rm Mean}(h)f - \cI_\Theta(h)f}{h} = 0,
	\]
	where $\Theta:=\{b\theta\, |\, b\geq 0\}$ with $\theta(x):=\nabla f(x)-\int_\X \nabla f(y)\, \mu(\d y)$ for all $x\in \X$. Moreover,
	\begin{equation} \label{eq.derivative.mean}
		\lim_{h \downarrow 0} \frac{\cI^{\rm Mean}(h) f - \mu f}{h} = \phi^* \left( \left( \int_\X \left\| \nabla f(y) - \int_\X \nabla f(y)\,\mu(\d y) \right\|^2 \, \mu(\d y) \right)^{1/2} \right).
	\end{equation}
\end{theorem}

Equation \eqref{eq.derivative.mean} appears in a different form in \cite{Bartl2021sensitivityAppendix}. There, the authors consider general linear constraints, whereas we focus on the mean constraint and provide an asymptotic formula for more general penalty functions. In that sense, Theorem \ref{thm.mean.constraint} is a partial extension of \cite[Theorem 7.7]{Bartl2021sensitivityAppendix}.

\subsection{Martingale constraint}\label{sec.main.martingale}

We now study an extension of the first order expansion given by Equation \eqref{eq.main.itheta}. In particular, we show how the inclusion of a martingale constraint improves the rate of convergence, leading to a \textit{second order} parametrization of the operator $\cI(h)$, for infinitesimally small $h>0$. This parametrization will no longer be given in terms of Monge transports of the baseline measure $\mu$, but by an actual randomization via an additional coin flip.

We start by introducing the notion of a \textit{martingale coupling}, which is of fundamental importance for this section. Let $\nu\in \cP_p$ and $\pi\in \cpl(\mu,\nu)$. Then, $\pi$ is called a \textit{martingale coupling} between $\mu$ and $\nu$ if
\[
\int_{\X\times \X} z\cdot \id_B(y)\, \pi(\d y,\d z)= \int_{\X} y\cdot \id_B(y)\, \mu(\d y)\quad\text{for all }B\in \mathcal B(\X).
\]
The set of all $\nu\in \cP_p$, for which there exists a martingale coupling between $\mu$ and $\nu$, is denoted by $\cM_p(\mu)$.

In the following, we use the notation $L(\X)$ for the space of bounded linear operators $\X\to \X$ and we denote by $\|\cdot\|=\|\cdot\|_{L(\X)}$ the usual operator norm on this space.
Moreover, given a bounded symmetric (or, equivalently, bounded self-adjoint) operator $S \in L(\X)$, we define
\begin{equation} \label{eq.maximal.curvature}
	|S|_\Max := \sup_{\|w\| \le 1} \langle w, Sw \rangle. 
\end{equation}
Note that if $S$ negative semidefinite, $|S|_\Max = 0$. Moreover, $|S|_\Max\leq \|S\|$ for all $S\in L(\X)$, and the supremum in \eqref{eq.maximal.curvature} can be taken over all $w\in \X$ with $\|w\|\in \{0,1\}$.

Throughout this section, we assume that the penalty function satisfies
\begin{equation} \label{eq: property of phi_v2}
	\liminf_{v \to \infty} \frac{\phi(v)}{v^{p/2}} > 0.
\end{equation}
Again, this property implies that the convex conjugate $\phi^*\colon [0,\infty)\to [0,\infty)$, given by \eqref{eq.def.conjugate}, is well-defined, convex, and continuous with $\phi^*(0)=0$.

We call $C^2_p$ the space of all twice continuously (Fr\'echet) differentiable functions $f \colon \X \to \R$, for which there exists a constant $C\geq 0$ such that
\[
\| \nabla^2 f(x) \| \le C(1 + \|x\|)^{p-2}\quad \text{for all }x\in \X.
\]
Note that $C_p^2$ is a subspace of $C_p^1$, in particular, it is contained in $\Lip_p$.

We include the martingale constraint in our study by considering a slightly different version of the operator $\cI(h)$ that we call $\cI^\text{Mart}(h)$ for $h>0$.\ To that end, we define the \textit{$p$-martingale Wasserstein distance} between $\mu$ and $\nu \in \cM_p(\mu)$ as
\[
\cW_p^\text{Mart}(\mu, \nu):=\left( \inf_{\pi \in \mart(\mu, \nu)} \int_{\X \times \X} \|z - y\|^p\,\pi(\d y, \d z) \right)^{1/p},
\]
where $\mart(\mu, \nu)$ is the set of all martingale couplings between $\mu$ and $\nu$.
Then, for every function $f \in \Lip_p$ and every $h > 0$, we define
\begin{equation} \label{eq.def.IMart}
	\cI^\text{Mart}(h)f := \sup_{\nu \in \cM_p(\mu)} \left( \int_\X f(z)\, \nu(dz) - \phi_h\big(\cW_p^\text{Mart}(\mu, \nu)^2\big) \right),
\end{equation}
where, again, $\phi_h(v) = h \phi(\frac{v}{h})$ for all $v\in [0,\infty)$.\ We point out that, apart from the additional martingale constraint in the Wasserstein distance, the penalty function is now applied to the squared (martingale) Wasserstein distance between $\mu$ and $\nu$, whereas, in the previous sections, it was applied to the (usual) Wasserstein distance without an additional power.

Lemma \ref{lem.restriction.2nd} shows that the growth condition \eqref{eq: property of phi_v2} allows us to restrict the supremum in \eqref{eq.def.IMart} to a ball of radius $\sqrt{ah}$, for a suitable constant $a \geq  0$ and $h>0$ sufficiently small, so that, in this case, uncertainty grows proportionally to the square root of $h$.

As in the unconstrained case, we are looking for a suitable parametric version of the operator $\cI^\text{Mart}(h)$ for $h>0$.\
As in the previous subsections, let $\Theta \subset L_p(\mu;\X)$. For $\theta\in \Theta$, we consider the probability measure $\mu_\theta^\text{Mart} \in \cM_p(\mu)$, which is given by
\begin{equation}
	\int_\X f(z)\,\mu_\theta^\text{Mart}(\d z):=\int_\X \frac{f\big(y+\theta(y)\big)+f\big(y-\theta(y)\big)}{2}\,\mu(\d y)\quad \text{for all }f\in \Lip_p.
\end{equation}
Note that, by definition,
\begin{equation}
	\int_\X f(z)\,\mu_\theta^\text{Mart}(\d z)=\int_\X \int_\R f\big(y+s\theta(y)\big)\,B_{\sym}(\d s)\,\mu(\d y)\quad \text{for all }f\in \Lip_p,
\end{equation}
where $B_{\sym}$ is the symmetric Bernoulli distribution on $\R$ with equal probabilities, i.e., $$B_{\sym}\big(\{-1\}\big)=B_{\sym}\big(\{1\}\big)=\frac{1}2.$$

For $h>0$, we define the \textit{parametric version} $\cI_\Theta^\text{Mart}(h)$ of $\cI^\text{Mart}(h)$ by
\begin{equation} \label{eq.I.theta.mart}
	\cI_\Theta^\text{Mart} (h)f:=\sup_{\theta\in\Theta }\bigg(\int_\X f(z)\, \mu_\theta^\mart(\d z)-\phi_h\big(\| \theta \|_{L_p(\mu; \X)}^2\big)\bigg)\quad \text{for all }f\in \Lip_p.
\end{equation}
Choosing a perfectly correlated coupling, we find that
\begin{align*}
	\cW_p^\text{Mart}\big(\mu, \mu_\theta^\mart\big)^2& \le \left( \int_\X \int_\R \|s\theta(y)\|^p\,B_{\sym}(\d s)\,\mu(\d y)\right)^{2/p}\\
	&= \left( \int_\X \|\theta(y)\|^p\,\mu(\d y)\right)^{2/p}=\|\theta\|_{L_p(\mu;\X)}^2.
\end{align*}
In particular, $\cI_\Theta^\text{Mart}(h)f \le \cI^\text{Mart}(h)f$ for all $h> 0$ and $f\in \Lip_p$.

We now state the convergence result for the operator $\cI^\text{Mart}(h)$ as $h\downarrow 0$. The proof is again relegated to Section \ref{sec.proof.main}.

\begin{theorem} \label{thm.martingale.constraint}
	Assume that $\Theta$ is dense in $L_p(\mu;\X)$ with $0\in \Theta$.
	Then,
	\[
	\lim_{h \downarrow 0} \frac{\cI^\mart(h)f - \cI^\mart_\Theta(h)f}{h} = 0 \quad \text{for all } f \in C_p^2.
	\]
	Moreover, for all $f \in C_p^2$,
	\begin{equation} \label{eq.derivative.mart}
		\lim_{h \downarrow 0} \frac{\cI^\mart(h)f - \mu f}{h} = \phi^*\left( \frac{1}{2} \left( \int_\X \big(\big|\nabla^2 f(y)\big|_\Max\big)^\frac{p}{p-2} \, \mu(\d y) \right)^\frac{p-2}{p} \right).
	\end{equation}
\end{theorem}

\begin{remark} \label{rem.mart.meas.select}
	The proof of the previous theorem is based on the fact that, for every function $f \in C^2_p$ and every $\eps > 0$, there exists a measurable function $v \colon \X \to \X$ with $\|v(x)\| \in \{0,1$\} and 
	\begin{equation}\label{eq.measurable curv}
		\big|\nabla^2 f(x)\big|_\Max^{\frac2{p-2}}\big\langle v(x), \nabla^2 f(x) v(x) \big\rangle \ge \big|\nabla^2 f(x)\big|_\Max^{\frac{p}{p-2}} - \eps\quad \text{for all }x\in \X.
	\end{equation}
	In fact, let $\eps>0$ and $\mathbb S_0:=\{w\in \X\, |\, \|w\|=1\}\cup\{0\}$ denote the unit sphere enriched by the neutral element in $\X$. Then, the function $\X \times \mathbb S_0 \to \R,\; (x,w) \mapsto \frac{1}{2} \langle w, \nabla^2 f(x) w \rangle$ is continuous, and we can apply a measurable selection argument, see \cite[Proposition 7.34]{bertsekas2004stochastic}, in order to come up with a measurable function $v\colon \X\to \mathbb S_0$ satisfying \eqref{eq.measurable curv}.
	If $\X = \R^d$, the set $\S_0$ is compact, and the supremum in Equation \eqref{eq.maximal.curvature} is indeed a maximum. In this case, we can invoke \cite[Proposition 7.33]{bertsekas2004stochastic} to obtain a \textit{measurable direction of maximal curvature} $v$, i.e., a measurable map $v\colon \X\to \X$ with $\|v(x)\|\in \{0,1\}$ and
	\[
	\big\langle v(x), \nabla^2 f(x) v(x) \big\rangle = \big|\nabla^2 f(x)\big|_\Max\quad \text{for all }x\in \X.
	\]
	Inspecting the proof of the previous theorem, we immediately get the following reduction of the parameter set if $f$ has a measurable direction of maximal curvature.
\end{remark}

\begin{corollary}
	Assume that $f \in C^2_p$ has a measurable direction of maximal curvature $v$, and define
	\[ \theta(x)=v(x)(|\nabla^2 f(x)|_\Max)^{1/(p-2)} \quad \text{for all } x \in \X. \]
	Then,
	\[
	\lim_{h \downarrow 0} \frac{\cI^\mart(h)f - \cI^\mart_\Theta(h)f}{h} = 0,
	\]
	whenever $\{ \lambda \theta \,|\, \lambda \ge 0 \} \subseteq \Theta$.
\end{corollary}

\section{Numerical computation of the parametric risk functional} \label{sec.numerics}

In this section, we discuss numerical methods to approximate the functional $\cI_\Theta(h)$ for $h>0$.\ We show how, under different hypotheses, the infinite-dimensional optimization problem related to the computation of $\cI_\Theta(h)$ can be reduced to a finite-dimensional optimization for all $h>0$. This echoes and complements the results in \cite{esfahani2018data} and \cite{ecksteinNN2021}.

We distinguish between two relevant cases: the one where the baseline measure $\mu$ is a convex combination of Dirac measures, and the more general one, which also includes regular measures. In the first case, the operator $\cI_\Theta(h)$ can be obtained through a finite-dimensional optimization scheme, whereas in the second case we approximate the operator $\cI_\Theta(h)$ via a neural network approach for all $h>0$. Moreover, we show how to extend these numerical schemes to the constrained parametric operators.

\subsection{Convex combination of Dirac measures} \label{sec.numerics.diracs}
Throughout this subsection, we assume that $\mu = \sum_{i = 1}^n a_i \delta_{x_i}$, where, for $i\in \{1,\ldots, n\}$, $\alpha_i\in (0,\infty)$ and $\delta_{x_i}$ is the Dirac measure centered in $x_i \in \X$. Moreover, since $\mu$ is a probability measure, $\sum_{i = 1}^n \alpha_i = 1$ is implicitly assumed. Such a type of baseline model is particularly relevant when considering empirical distributions, including data driven problems where one observes historical events and assumes that the generating distribution is the uniform distribution on the sample points (i.e., $a_i=1/n$ for all $i\in \{1,\ldots, n\}$). In this case, the worst case loss, represented by $\cI(h)f$ with uncertainty level $h>0$, accounts for statistical errors present in the sample, cf.\ \cite{esfahani2018data} for more details and real world applications of this type of estimates.

In this situation, two functions $g_1, g_2 \colon \X \to \X$ are $\mu$-a.s.\ equal if and only if they coincide on the atoms of the distribution $\mu$, i.e., if and only if $g_1(x_i) = g_2(x_i)$ for all $i \in\{ 1, \dots, n\}$. In particular, the set $L_p(\mu,\X)$ is isomorphic to the $n$-fold Cartesian product $\X^n:=\Pi_{i=1}^n \X$ of $\X$ with itself (endowed with the product topology).\ This leads to a drastic simplification of the related optimization problem, and if $\X$ is finite-dimensional, so is the parameter set. We summarize these observations in the following proposition.

\begin{proposition} \label{prop.dirac.optim}
	Let $n\in \N$, $x_1,\ldots, x_n\in \X$, $\alpha_1,\ldots\alpha_n\in (0,\infty)$ with $\sum_{i=1}^n \alpha_i = 1$, and
	\[ \mu = \sum_{i=1}^n \alpha_i \delta_{x_i}.\]
	Moreover, let $\mathcal D\subset \X^n$ be a dense subset of $\X^n$ with $0\in \mathcal D$, and define
	\[
	\cI_{\mathcal D}(h)f:=\sup_{(\theta_1, \dots, \theta_n) \in \mathcal D} \left( \sum_{i = 1}^n \alpha_i f(x_i + \theta_i) - \phi_h \left(\left(\sum_{i=1}^n \alpha_i \|\theta_i\|^p\right)^{1/p}\right) \right)\quad\text{for all }f\in \Lip_p.
	\]
	Then, 
	\[
	\lim_{h \downarrow 0} \frac{\cI(h)f - \cI_D(h)f}{h} = 0
	\]
	for every $f\in \Lip_p$ with a measurable direction of steepest ascent.
\end{proposition}
Note that, in case of $\X=\R^d$, this result is somewhat similar to \cite[Theorem 4.4]{esfahani2018data}, where, working with an empirical distribution, the authors show that the worst case loss, for a function that is the pointwise maximum of concave functions (and satisfies some other conditions), can be computed via a convex shifting of the initial sample points.
In our framework, we can widely relax the hypothesis on the loss function $f$, consider more general atomic measures, and come up with a result in the spirit of \cite[Theorem 4.4]{esfahani2018data} for infinitesimally small amounts of uncertainty. In conclusion, if the baseline measure is given by a convex combination of Dirac measures, for small levels of uncertainty, one does not have to spread away mass in order to find the worst case loss.

\subsection{Absolutely continuous measures and approximation by neural networks}
\label{sec.numerics.abscont}
In this subsection, we provide a framework to build an approximation of the operator $\cI_\Theta(h)$, for $h>0$, in the case where $\X = \R^d$ and the baseline measure $\mu$ is absolutely continuous with respect to the Lebesgue measure.\ In this case, the set $L_p(\mu; \X)$ cannot be reduced to a finite-dimensional space. We therefore explore a numerical procedure, based on a neural network approach, to approximate the operator $\cI_\Theta(h)$ for $h>0$.

Following the method proposed in \cite{ecksteinNN2021}, the idea is to consider an increasing family of parameter sets, which are finite-dimensional and whose union is dense in $L_p(\mu; \X)$.

The numerical method builds on the following result, which holds true for every baseline measure $\mu$, even if we apply it mainly to reference measures which are absolutely continuous with respect to the Lebesgue measure.

\begin{proposition} \label{prop.dense.approx}
	Let $(\Theta^n)_{n\in \N}$ be a family of subsets of $L_p(\mu;\X)$ with $\Theta^n \subseteq \Theta^{n + 1}$ for all $n \in \N$. If the set $\bigcup_{n \in \N} \Theta^n$ is dense in $L_p(\mu;\X)$ and contains $0$, then
	\[
	\cI_{\Theta^n} (h) f\nearrow   
	\cI_{L_p(\mu;\X)} (h)f \quad\text{as }n\to \infty
	\]
	for all $h>0$ and $f \in \Lip_p$.
\end{proposition}

\begin{proof}
	Let $h > 0$, $f \in \Lip_p$, and $\Theta:= \bigcup_{n \in \N} \Theta^n$.\ The inclusion $\Theta^n \subseteq \Theta^{n+1}$ implies that $\cI_{\Theta^n}(h)f\leq \cI_{\Theta^{n+1}}(h)f$ for all $n\in \N$. Moreover,
	$$
	\cI_{\Theta} (h)f= \sup_{n\in \N}\cI_{\Theta^n} (h) f\leq  \cI_{L_p(\mu;\X)} (h)f.
	$$
	The statement now follows from Lemma \ref{lem.approx}.
\end{proof}

Although the set $L_p(\mu; \X)$ might, a priori, be very large, in particular if the reference measure $\mu$ is absolutely continuous with respect to the Lebesgue measure, Proposition \ref{prop.dense.approx} tells us that we can implement a numerical scheme to approximate $\cI_{L_p(\mu;\X)}(h)$ for all $h>0$ based on an increasing sequence of subsets $\Theta^n$ of $L_p(\mu;\X)$, for which we can compute $\cI_{\Theta^n}(h)$.

Considering the case $\X = \R^d$, \cite[Theorem 1]{HORNIK1991251} ensures that the set of neural networks with one layer and arbitrary number of neurons is dense in $L_p(\mu)$ and, with a simple generalization, in $L_p(\mu;\R^d)$, cf.\ \cite[Corollary 2.6 and Corollary 2.7]{HORNIK1989359}). Therefore, using Proposition \ref{prop.dense.approx}, we can can approximate the quantity $\cI_{L_p(\mu;\X)}(h) f$ with the quantity $\cI_{\Theta^n}(h)f$ for all $h>0$, where $\Theta^n$ is the set of functions from $\R^d$ to $\R^d$ given by neural networks with one layer and $n$ neurons, for $n\in \N$ sufficiently large.\ We provide a more precise framework, following closely what has been done in \cite{ecksteinNN2021}.

For $m\in \N$, we consider fully connected feed-forward neural networks from $\R^d$ to $\R^d$, which are maps of the form
\[
x \mapsto A_m \circ \psi \circ A_{m - 1} \circ \cdots \circ \psi \circ A_0(x),
\]
where $A_i$ are affine transoformations of the form $A_i = M_ix + b_i$, for a matrix $M_i$ and a vector $b_i$, $\psi \colon \R \to \R$ is a nonlinear activation function, and the application of the \textit{activation function} is to be understood componentwise, i.e., $\psi((x_1, \dots, x_d)) = (\psi(x_1), \dots, \psi(x_d))$. We say that the neural network has \textit{depth} $m$ and \textit{width} $n\in \N$ (alternatively has $m$ \textit{layers} and $n$ \textit{neurons} per layer) if $A_0$ is a map from $\R^d$ to $\R^n$, $A_1, \dots, A_{m-1}$ are maps from $\R^n$ to $\R^n$, and $A_m$ is a map from $\R^n$ to $\R^d$. The coefficients of the matrices and the vectors forming the affine transformations $A_1,\ldots, A_m$ represent the parameters of the network and can be regarded as an element of $\R^N$, for some $N \in \N$ that depends on the structure of the network. We call $\fN_d^{m,n}$ the set of neural network from $\R^d$ to $\R^d$ with depth $m$ and width $n$.
In case we want to specify the activation function used in the network, we also write $\fN_d^{m,n}(\psi)$, where $\psi$ is the activation function. Moreover, $\fN_d^m := \bigcup_{n \in \N} \fN_d^{m,n}$ denotes the set of neural networks with $m$ layers and arbitrary number of neurons.\ Proposition \ref{prop.dense.approx} can now be restated in a neural network framework.

\begin{corollary} \label{cor.nn.approx}
	For every integer $m \in\N$ and every bounded, nonconstant, Lipschitz continuous activation function $\psi$, we have
	\[
	\cI_{\fN_d^{m,n}}(h) f\nearrow \cI_{L_p(\mu;\X)} (h) f \quad\text{as }n\to \infty 
	\]
	for all $h > 0$ and $f \in \Lip_p$.
\end{corollary}
\begin{proof}
	The convergence for $m = 1$ is direct consequence of Proposition \ref{prop.dense.approx} and \cite[Theorem 1]{HORNIK1991251}. Since the activation function $\psi$ is Lipschitz continuous, the functions in $\fN_d^{m,n}$ are also Lipschitz continuous as a composition of Lipschitz continuous functions, so that $\fN_d^{m,n} \subseteq L_p(\mu; \X)$. Moreover, $\fN_d^{1,n} \subseteq \fN_d^{m,n}$, and therefore
	\[
	\cI_{\fN_d^{1,n}}(h)f \le \cI_{\fN_d^{m,n}}(h)f \le \cI_{L_p(\mu;\X)}(h)f\quad\text{for all }n\in \N.
	\]
\end{proof}

\begin{remark} \label{rem.nn.approx}
	\begin{enumerate}[a)]
		\item \label{rem.nn.approx.full} As already stated in Remark \ref{rem.abscont.param}, when the baseline measure $\mu$ is absolutely continuous w.r.t.\ the Lebesgue measure, the parametrized version $\cI_{L_p(\mu;\X)}(h)f$ corresponds in fact to $\cI(h)f$, so that the previous corollary gives us an approximation scheme for $\cI(h)f$ for all $h>0$.
		\item The fact that we get an approximation from below complements the approach in \cite{ecksteinNN2021}, where the authors obtain an approximation from above based on a version of the Daniell-Stone theorem, which leads to a dual representation in terms of a superhedging problem. Therefore, their functional is expressed as the infimum over a set of functions that they approximate via neural networks.\ Combining the two approaches, one can overcome the problem of an unknown rate of convergence as $h\downarrow 0$ by estimating the relevant quantity from below and from above. In practice, one can then train the two algorithms until the gap between the two estimates is below a certain threshold.
		\item The assumption of a bounded activation function can be relaxed, and we can allow for unbounded Lipschitz continuous functions using more layers, cf.\ \cite[Section 3]{HORNIK1991251} and \cite[Remark 3.4]{ecksteinNN2021}. In fact, let $\zeta$ be a bounded, nonconstant, and Lipschitz continuous activation function, and let $\psi$ be an unbounded activation function such that $\zeta \in \fN_1^1(\psi)$. Then, for $m>1$, we have $\fN_d^m(\zeta) \subset \fN_d^m(\psi)$ and hence $\fN_d^m(\psi)$ is dense in $L_p(\mu; \R^d)$, since $\fN_d^m(\psi) \subseteq L_p(\mu; \R^d)$ due to the Lipschitz continuity of $\psi$. This allows us to work, for example, with ReLU activation functions, i.e., $\psi(x) = \max\{x, 0\}$, since one can obtain the bounded function $\zeta(x) = \min\{\max\{x,0\},1\}$ from a ReLU network. Alternatively, note that we can prove Corollary \ref{cor.nn.approx} for neural networks with ReLU activation function directly using the simple approximation result, cf.\ Proposition \ref{prop.append.approx}, without invoking the universal approximation of \cite[Theorem 1]{HORNIK1991251}.\ In fact, a single layer network with ReLU activation function and arbitrary number of neurons can reproduce the span of the set $L_1$ in Remark \ref{rem.append.approx.ex} c).
		\item One could also choose a different class of universal approximators, like polynomials. Our choice of employing a neural network framework comes from their recent applications to optimal transport and penalized problems, as, for example, in \cite{eckstein2020robustrisk} and \cite{ecksteinNN2021}. Moreover, in Section \ref{sec.example.financial}, we shall give a financial interpretation to the neural network approximation.
	\end{enumerate}
\end{remark}

In the numerical implementation of the algorithm, we train the network using as loss function
\[
L(\theta):=- \left( \int_{\R^d} f(z)\,\mu_\theta(\d z) - \phi_h(\|\theta\|_{L_p(\mu; \R^d)}) \right)\quad\text{for all }\theta\in \Theta,
\]
where the integrals are computed via Monte Carlo simulations. Note that this corresponds to a classical stochastic gradient descent, since at each step of the optimization new samples are drawn from the distribution $\mu$.

Finally, we remark that this approach can also be employed when working with empirical distributions.\ In particular, when the sample size is high, stochastic gradient descent might be preferred over the full-batch finite-dimensional optimization proposed in Section \ref{sec.numerics.diracs}, see, for example, \cite{Bottou2010}.

\subsection{Extension of neural network scheme to constrained settings}  \label{sec.numerics.constraints}
In this section, we briefly explain how to extend the neural network scheme proposed in Section \ref{sec.numerics.abscont} to the settings with mean and martingale constraint, cf.\ Section \ref{sec.main.mean.constr} and Section \ref{sec.main.martingale}.

In the case of a mean constraint, one should in principle restrict the optimization to functions in $L_p(\mu;\X)$ with mean zero. Even if, a priori, it is not possible to impose this constraint on the neural network, we can obtain it simply by subtracting componentwise the Monte Carlo mean of the vector field. Therefore, at each step of the optimization, we only have to perform the computation of the quantity $ \int_\X \theta(y)\,\mu(\d y)$, which doesn't require any additional sampling.

The martingale constraint almost doesn't require any additional computational effort. As a matter of fact, in this framework we have
\[
\int_\X f(z)\,\mu_\theta(\d z) = \int_\X \frac{f\big(y + \theta(y)\big) - f\big(y - \theta(y)\big)}{2}\,\mu(\d y)\quad \text{for all }\theta\in L_p(\mu;\X),
\]
and it is clear that we only have to perform one evaluation more in each step of the optimization.

\section{Examples and applications} \label{sec.examples}
The aim of this section is to illustrate the results developed in the previous sections in several examples and applications.

\subsection{A toy model for an earthquake}\label{sec.ex.earthquake}
We start by investigating a toy model for an earthquake, i.e., a one-period model, where the baseline measure $\mu$ describes an a priori estimate for the location of the epicenter of an earthquake in $\R^2$.
Suppose that an insurance company, selling a coverage for damages caused by earthquakes, wants to compute the worst case loss in terms of uncertainty regarding the precise location of the epicenter. For the sake of simplicity and analytical tractability, we choose a loss function that is given by the following compactly supported version of a Gaussian kernel:
\[
f(x_1,x_2) = \begin{cases}
	c\exp \left( 1 + \frac{r^2}{(x_1 - \bar{x}_1)^2 + (x_2 - \bar{x}_2)^2 - r^2}\right), & \text{if } (x_1 - \bar{x}_1)^2 + (x_2 - \bar{x}_2)^2 < r^2, \\
	0, & \text{otherwise,}
\end{cases}
\]
where $\bar{x} \in \R^2$ can, for example, represent the location of a city center, $c>0$ and $r>0$ are positive constants representing, respectively, the maximal loss level and the radius of the area covered by the insurance. The idea is that losses decrease as the epicenter of the earthquake moves away from the city center. When needed, we will write $f(c, r, \bar{x}_1, \bar{x}_2; x_1, x_2)$ to highlight the dependence on $c$, $r$, and $\bar{x}$. Figure \ref{fig.loss.earthquake} depicts this loss function for $r = 1.5$, $c = 1$, and the city center $\bar x$ located in the origin.

\begin{figure}
	\centering
	\includegraphics[scale=0.65]{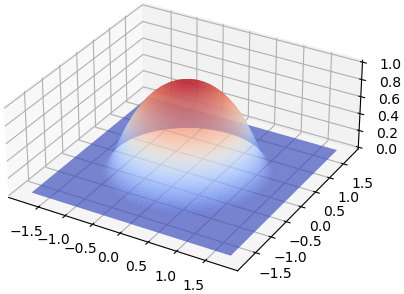}
	\caption{Loss function for the earthquake model}
	\label{fig.loss.earthquake}
\end{figure}

In the following, we work with the $2$-Wasserstein distance, and numerically approximate the quantity
\[
\cI(h)f=\sup_{\nu \in \cP_2} \left( \int_{\R^2} f(z_1,z_2)\,\nu(\d z_1, \d z_2) - \phi_h (\cW_2(\mu, \nu)) \right)\,
\]
computing its parametric version
\[
\cI_{L_2(\mu;\R^2)}(h)f=\sup_{\theta \in L_2(\mu;\R^2)} \left( \int_{\R^2} f(z_1,z_2)\,\mu_\theta(\d z_1, \d z_2) - \phi_h \big(\|\theta\|_{L_2(\mu; \R^2)}\big) \right)\,
\]
through a procedure which is independent of the loss function $f$. We consider the penalty function $\phi(v) = v^2$ for $v\geq 0$, and we explore two different frameworks, one leading to a finite-dimensional optimization and another one based on neural networks.

\subsection{Epicenter as a sum of Dirac measures}
Consider the loss function from Section \ref{sec.ex.earthquake} and a baseline model given by
\[ \mu = \frac{1}{3} \delta_{x_1} + \frac{1}{3} \delta_{x_2} + \frac{1}{3} \delta_{x_3} \]
with $x_1 = (0.55 ,\, 0)$, $x_2 = (0    ,\, 0.85)$, and $x_3 = (-1.10,\, 0)$.\ The optimization suggested by Proposition \ref{prop.dirac.optim} can be easily performed with \emph{out of the box} tools that are nowadays integrated in many scientific programming languages.\
We implement the optimization using Python 3.9\footnote{All source codes are available at \url{https://github.com/sgarale/wasserstein_parametrization}.} and the optimizer included in the package Scipy (version 1.7.3), which applies the Quasi-Newton method BFGS, see \cite[Chapter 6]{Nocedal2006}.\
Figure \ref{fig.dirac.optim} illustrates the initial situation, the optimizers, and the value of the operator $\cI_{L_2(\mu;\R^2)}(h)$ at different uncertainty levels $h$ ranging from $0$ to $0.3$. Looking at Figure \ref{fig.dirac.optim}  (b), we see how the worst case loss is obtained by shifting the three possible locations for the epicenter of the earthquake towards the city center, which means in the direction of the gradient of the loss function. This is in line with Theorem \ref{cor.main.itheta.single} and Remark \ref{rem.meas.direc}, nevertheless it is worth noting that the optimization procedure is independent of the particular loss function. Moreover, note that, because of the penalization term, one has to pay a price to move the atoms of the measure; this is why the worst case loss is obtained by shifting more the points whose incremental contribute to the loss is higher, that is, those points which lie in a region where the loss function is steeper.

\begin{figure}
	\centering
	\subfloat[]{\includegraphics[width=0.40\textwidth]{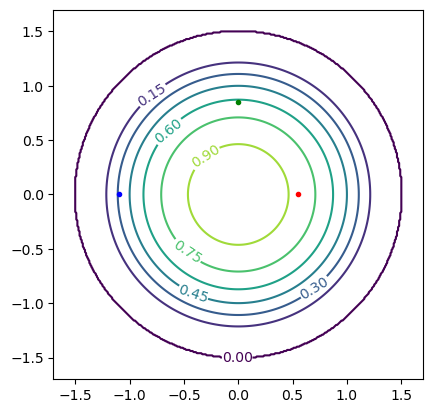}} \quad
	\subfloat[\label{fig.dirac.optim.optimizers}]{\includegraphics[width=0.40\textwidth]{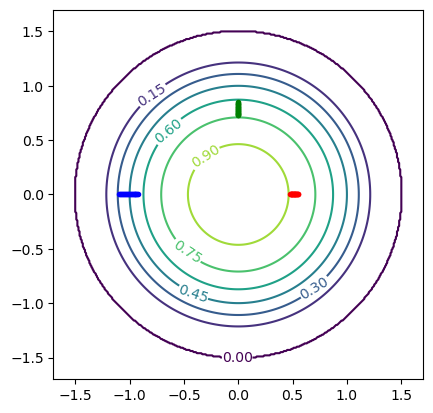}} \\
	\subfloat[]{\includegraphics[width=0.40\textwidth]{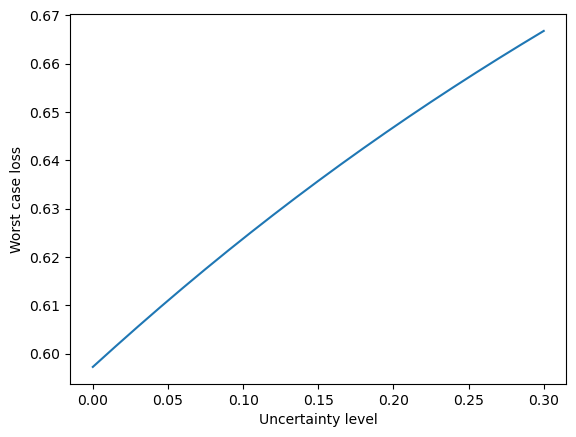}}
	\caption{Calculation of the worst case loss for a reference model given by a convex sum of Dirac measures. Figure (a): loss function and atoms of the baseline measure.\ Figure (b): optimizers as the uncertainty level increases.\ Figure (c): worst case loss for different levels of uncertainty.}
	\label{fig.dirac.optim}
\end{figure}

\subsection{Epicenter with absolutely continuous distribution} \label{sec.example.abscont}
Here, we exploit the neural network approach developed in Section \ref{sec.numerics.abscont} to approximate the value of $\cI_\Theta(h)f$, for $h>0$, in the case where the baseline measure $\mu$ is a two-dimensional Gaussian random variable with mean $(1,\,0)$ and covariance matrix equals the identity matrix. Again, we implement the code in Python 3.9, using the package Pytorch \cite{pytorch2019} (version 1.11.0) to handle the neural network architecture.

\begin{figure}
	\centering
	\subfloat[Uncertainty $0.002$]{\includegraphics[width=0.40\textwidth]{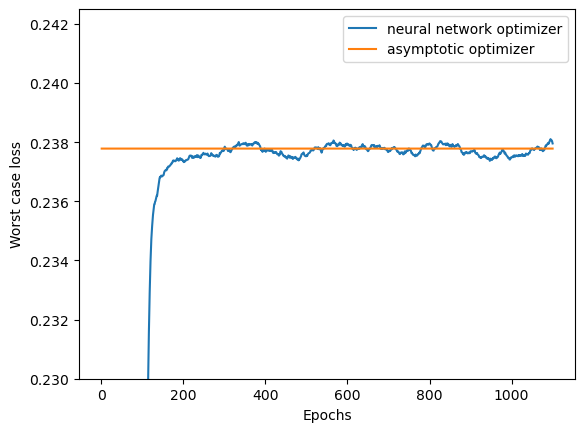}} \quad
	\subfloat[Uncertainty $0.005$]{\includegraphics[width=0.40\textwidth]{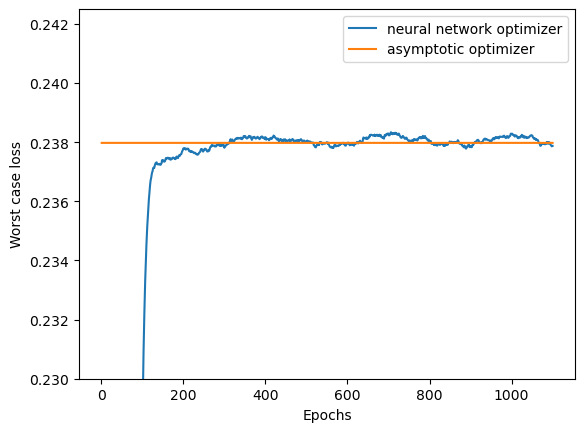}} \\
	\subfloat[Uncertainty $0.02$]{\includegraphics[width=0.40\textwidth]{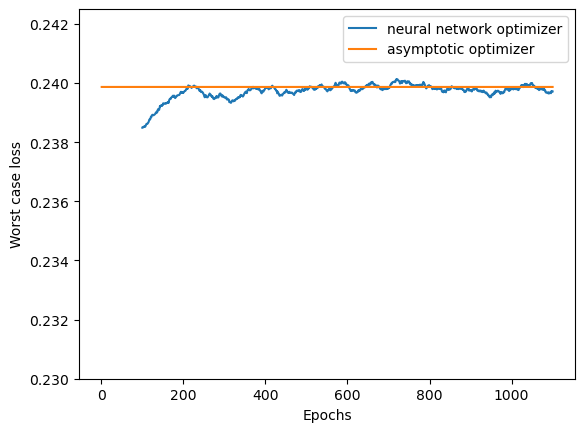}} \quad
	\subfloat[Uncertainty $0.5$ \label{fig.gaussian.optim.training.0.5}]{\includegraphics[width=0.40\textwidth]{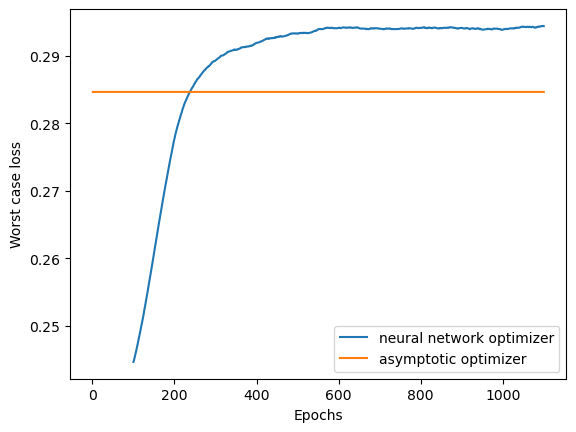}}
	\caption{Traning of neural network for the calculation of the worst case loss for a reference model given by a Gaussian random variable. The yellow line in all the graphs represents the result of the one-dimensional optimization performed using the gradient of the loss function.\ All values are moving averages on a window of 100 epochs.}
	\label{fig.gaussian.optim.training}
\end{figure}

\begin{figure}
	\centering
	\subfloat[Uncertainty $0.02$]{\includegraphics[width=0.40\textwidth]{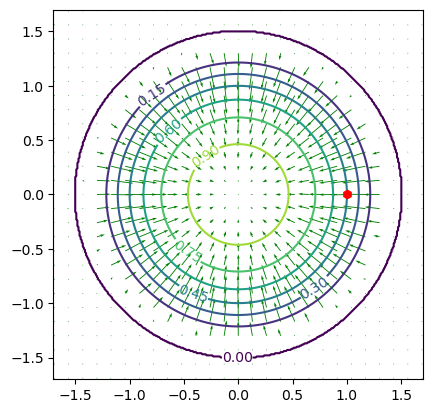}} \quad
	\subfloat[Uncertainty $0.5$]{\includegraphics[width=0.40\textwidth]{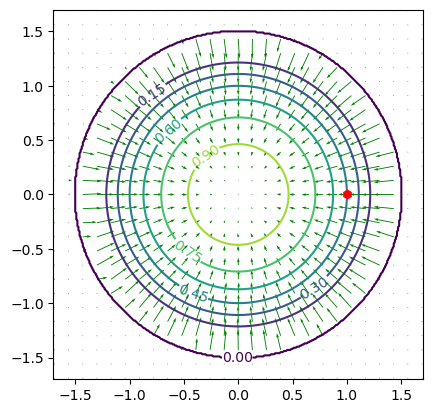}}
	\caption{Optimizers obtained by training a neural network to compute the worst case loss for the gaussian reference model. The contour plots represent the loss function and the red dots represent the location of the mean of the distribution of the epicenter.}
	\label{fig.gaussian.optim.optimizers}
\end{figure}

We use the built-in Adam optimizer, cf.\ \cite{Adam2014}, and present the results obtained from training a neural network with depth 4, width 20, and ReLU activation function $\psi(x)=\max\{x,0\}$. For the optimization, we use a learning rate of $0.001$
and a batch size of $32\,768$ samples (the numbers of points used to compute the integrals via Monte Carlo). The calculations are run on a laptop with Intel\textsuperscript{\textregistered}  Core\texttrademark \,i7-7600U processor (2.90GHz) and 16GB RAM. Computing the worst case loss at a fixed uncertainty level takes less than a minute.

Figure \ref{fig.gaussian.optim.training} shows the training phase of the algorithm at different uncertainty levels and Figure \ref{fig.gaussian.optim.optimizers} shows the optimizers obtained through the training.\ Theorem \ref{cor.main.itheta.single} tells us that, asymptotically as the level of uncertainty tends to zero, the optimizer is a scalar multiple of the gradient of the loss function (we are working with $p=2$), therefore we always compare the neural network optimization with the one-dimensional optimization computed using the gradient. Note how the optimizers resemble the gradient of the loss function for small levels of uncertainty, however, for higher values of $h$, for example $h = 0.5$, the neural network optimization outperforms the optimization along the gradient, see Figure \ref{fig.gaussian.optim.training} (d).

\subsection{An example of transfer learning} \label{sec.example.transferlearning}
We now apply the considerations of Remark \ref{rem.independ.measure} to our neural network framework. We have seen that, when the uncertainty $h>0$ is small, the optimizer for $\cI_\Theta(h)f$ is independent of the reference measure up to a one-dimensional optimization.\ We exploit this observation to provide a simplified framework to compute the worst case loss with respect to different reference measures.
Since, for small values of $h$, the optimizers for different reference measures differ only by a multiplicative factor, we add a multiplicative layer to the network; formally we add the composition with a matrix $M = c \id_{2 \times 2}$, where $\id_{2 \times 2}$ is the two-dimensional identity matrix and $c \in \R$ is the parameter of this additional layer. The network then has the structure
\[
x \mapsto M \circ A_m \circ \psi \circ A_{m - 1} \circ \cdots \circ \psi \circ A_0(x).
\]
Once we have trained the neural network on a starting reference measure, we keep the parameters of the affine transformations $A_0, \dots, A_m$ fixed, and allow only for the training of the last layer, which corresponds to performing a one-dimensional optimization.\ This drastically reduces the computation time of the worst case loss for different reference measures.

To avoid effects due to particular symmetries of the problem, we perform this exercise for the loss function
\[
\bar{f}(x_1, x_2) := f(0.5, 1, 0, 0; x_1, x_2) + f(0.3, 0.75, 1.25, 0; x_1, x_2),
\]
where the function $f$ is defined as in Section \ref{sec.ex.earthquake}.\
From an application perspective, one can think of two cities, a bigger one with the city center located in the origin and a smaller one with the city center located in $(1.25, 0)$.\ As reference measure, we consider a two-dimensional Gaussian random variable with mean located in $(0.75, 0.25)$ and the identity matrix as covariance matrix. Figure \ref{fig.double.dome.optimizer} depicts the training phase and the optimizer obtained by the training at the uncertainty level $h = 0.5$. Again, we observe that, for this level of uncertainty, the optimization of the full vector field outperforms the one-dimensional optimization obtained through the asymptotic optimizer.

\begin{figure}
	\centering
	\subfloat[Training phase]{\includegraphics[width=0.4\textwidth]{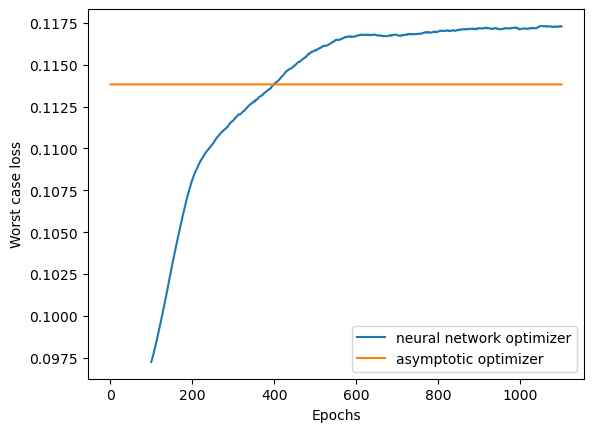}}
	\quad
	\subfloat[Optimizing vector field]{\includegraphics[width=0.4\textwidth]{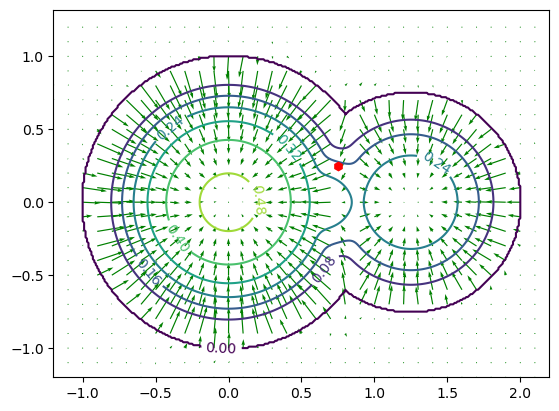}}
	\caption{Training phase and optimizer obtained for the loss function $\bar{f}$ in Section \ref{sec.example.transferlearning}. The red dot is the location of the mean of the distribution of the epicenter.}
	\label{fig.double.dome.optimizer}
\end{figure}

Figure \ref{fig.transfer.training} shows the training phase of the partial training for different reference measures, where the vector field obtained by the training on the initial reference measure is optimally rescaled, in comparison with the full training on each reference measure. In this example, we move the mean of the reference model, which stays a Gaussian random variable with identity matrix as covariance matrix. These results suggest a strong stability with respect to the uncertainty level, even if, as seen before, the asymptotic optimizer is outperformed by the full neural network training for a high value of the uncertainty. In machine learning jargon, this technique is called \emph{transfer learning}; a method that is used both to reduce computational costs and to avoid a large bias from a full training when only sparse data is available. We refer to \cite{Bozinovski2020} for a list of references on this topic.

\begin{figure}
	\centering
	\subfloat[Mean:\,~$(0.75,\,-0.25)$]{\includegraphics[width=0.40\textwidth]{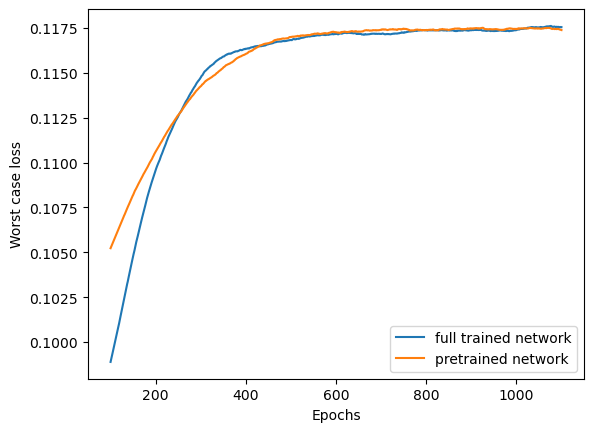}} \quad
	\subfloat[Mean:\,~$(-0.75,\,0.50)$]{\includegraphics[width=0.40\textwidth]{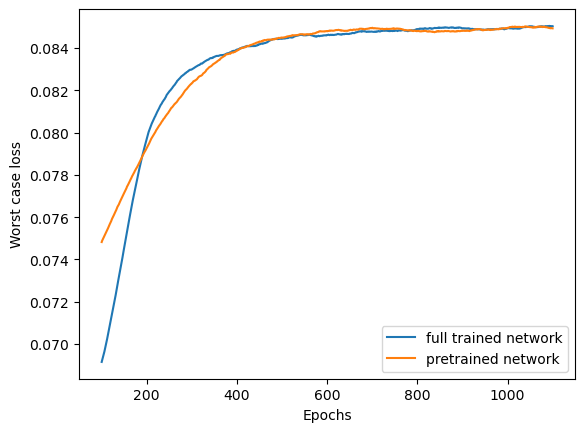}} \\
	\caption{Comparison of full training and partial training (with fixed inner layers) when the mean of the reference model is moved to different locations. Mean of the reference measure for the pretrained network: $(0.75,\,0.25)$.\ Full training batch size: $65\,536$. Partial training batch size: $4\,096$.}
	\label{fig.transfer.training}
\end{figure}

If only sparse empirical data is available, this observation suggests to train the neural network for the approximation of the worst case loss on a randomly generated and sufficiently large dataset, and then to fine-tune it with a one-dimensional training on the available empirical data.

\subsection{A financial perspective: replicating uncertainty via portfolios of call options or digital options} \label{sec.example.financial}
Here, we use a simple approximation result given in the Appendix \ref{app.A} to give a financial interpretation to the parametric approximation of the operator $\cI(h)$, for $h>0$, and to motivate the use of a neural network approximation.\
We assume that $\X = \R$, and consider the payoff function $f\colon \R\to \R$ of a financial contract depending on an underlying risk factor (typically an equity asset or index). The reference measure $\mu$ is then the distribution of this risk factor. Note that in order to include the case of an asset that takes only nonnegative values, one can simply pass to a log-scale, while maintaining the structure of a Hilbert space, cf.\ Section \ref{sec.example.martingale}.

Taking $\Theta = L_p(\mu;\R)$, we obtain a first order approximation of $\cI(h)$ as $h\downarrow 0$ by Theorem \ref{thm.main.itheta}. Moreover, Lemma \ref{lem.approx} tells us that we can rather look at a dense subset of $L_p(\mu;\R)$. As stated in Remark \ref{rem.append.approx.ex} c), for example, portfolios of call options and portfolios of digital options are dense in $L_p(\mu;\R)$. In fact, portfolios of digital options are the span of the set $L_0 := \{ \id_{[K,\infty)}\,|\,K \in \R \}$ and portfolios of call options are the span of the set $L_1:=\{[x \mapsto (x-K)^+]\,|\,K \in \R\}$. This means that we can obtain a first order approximation via the parametric version $\cI_\Theta(h)$, for small $h>0$, where the set $\Theta$ is given by the profiles of all possible portfolios of call options or digital options.

As we remarked in Section \ref{sec.main.mean.constr}, the hypothesis of absence of arbitrage, imposes a mean constraint on perturbations of the reference measure, which leads to so-called \textit{self-financing} portfolios.\ As a matter of fact, the condition $\int_\R \theta(y)\,\mu(\d z)=0$ corresponds to the fact that the portfolio with profile $\theta$ has price zero (we are working in a simple setting where the risk-free rate is zero, but this can be easily generalized), i.e., we can replicate the uncertainty on the price of a financial position without spending additional money.

As a final remark, we point out that an arbitrary portfolio of call options can be built using a single-layer neural network with ReLU activation function, which is exactly the profile of a call option with strike zero.\ Therefore, one can prove the density of ReLU networks using Proposition \ref{prop.append.approx} without invoking universal approximation theorems.\ Moreover, as a byproduct of the training, we get a portfolio of call options which replicates the uncertainty; this insight may have interesting implications for the financial industry.

\subsection{An example for the martingale constraint: pricing of a bull spread shortly after a quoted maturity} \label{sec.example.martingale}
We conclude this section by providing a framework that illustrates the implementation of the martingale constraint.

Suppose that we want to price a bull call spread option with strikes $K_1,K_2\in \R$ with $0 < K_1 < K_2$ and maturity $T+h$, for a small time $h >0$, having quoted options for the shorter maturity $T>0$. The payoff of the option is given by the sum of a long call option on the lower strike $K_1$ and a short call option on the upper strike $K_2$:
\[
f(x) = \max(x - K_1, 0) - \max(x - K_2, 0).
\]
This situation is complementary to a situation described in \cite{Cont2004}. There, it is shown that pricing a call option on a maturity that is shorter than the first quoted maturity, using a diffusion model with deterministic volatility curve, naturally leads to volatility uncertainty.
Here, we reverse the point of view, and assume that we know the distribution of the underlying at time $T$, and we look for the price of the option at time $T+h$. Since no information is provided by the market for the maturity $T+h$, we look for upper and lower price bounds based on Wasserstein uncertainty around the distribution $\mu$ of the underlying at time $T$, which can be seen as a best guess for the distribution of the underlying at time $T+h$. We assume $\mu$ to be a log-normal distribution coming from a Black-Scholes model, cf.\ \cite[Chapter 5]{FoellmerSchied}, with volatility $\sigma>0$ estimated from market data. This choice is quite natural since the log-normal distribution is used by the market to quote options via the implied volatility surface. For simplicity, we assume that the risk-free rate is zero, so that the value of the option at time $0$ for the maturity $T$ is given by
\[
V_{0,T}(f):=\int_\R f\big(e^y\big)\,N\big(-\tfrac{\sigma^2 T}2,\sigma^2 T\big)(\d y),
\]
i.e., we are working with the reference measure $\mu$  on $\R$, which is the distribution of $ e^Y$ for $Y\sim N\big(-\tfrac{\sigma^2 T}2,\sigma^2 T\big)$.
We make the following two assumptions:
\begin{enumerate}[a)]
	\item the uncertainty in the returns grows as the product of the square root of time and the volatility $\sigma$, i.e., the uncertainty level at time $T+h$ is $\sigma \sqrt{h}$,
	\item the market is free of arbitrage.
\end{enumerate}
The first hypothesis is in line with the assumption of a log-normal distribution as a baseline measure. In fact, embedding this distribution in a dynamic setting, we get a Black-Scholes model, where the volatility of the diffusion process increases as the square root of time; we are then imposing the same scaling for the uncertainty. The second hypothesis is standard and, in the spirit of the fundamental theorem of asset pricing, it implies that the asset process is a martingale, see, for instance, \cite[Theorem 1.7]{FoellmerSchied}.

Under these hypotheses, we can bound the value $V_{0,T+h}(f)$ of the option at time $0$ with maturity $T+h$ by
\[
-\cI^\text{Mart}(h)(-f) \le V_{0,T+h}(f) \le \cI^\text{Mart}(h)f\quad\text{for all }h>0,
\]
where the penalty function in the definition of the operator $\cI^\text{Mart}$ is given by $\phi = \infty \cdot\id_{(\sigma^2,\infty)}$.
Therefore, for small values of $h$, we can approximate these bounds estimating the operator $\cI_\Theta^\text{Mart}(h)f$.

Using an adjusted neural network scheme as explained in Section \ref{sec.numerics.constraints} and thinking of the Black-Scholes price as the price without uncertainty, we depict the outcome of this estimate in Figure \ref{fig.option.bounds} with reference time $T=1$ (one year) and intervals $h$ ranging from $0$ to $1/12$ (one month).\ We work with  $p=3$ and, in order to avoid numerical problems related to the value $\infty$ in the convex indicator function, we replace the penalization function $\phi= \infty \cdot\id_{(\sigma^2,\infty)}$ by the real-valued function $\phi(x) = \frac{1}{n} (x/\sigma^2)^n$ with $n=5$. Heuristically, this is coherent with our assumptions, since, thanks to Lemma \ref{lem.restriction.2nd}, we have the desired rescaling of the uncertainty up to a constant factor, and, moreover, this penalization converges pointwise to $\infty \cdot \id_{(\sigma^2,\infty)}$ as $n\to \infty$.\ Moreover, we remark that, although the payoff function of a bull spread is not twice continuously differentiable, at least on a discrete level, it is indistinguishable from a smoothened version of it, also in view of the fact that, in this example, the set of kinks has measure zero under $\mu$.
\begin{figure}
	\centering
	\includegraphics[width=0.50\textwidth]{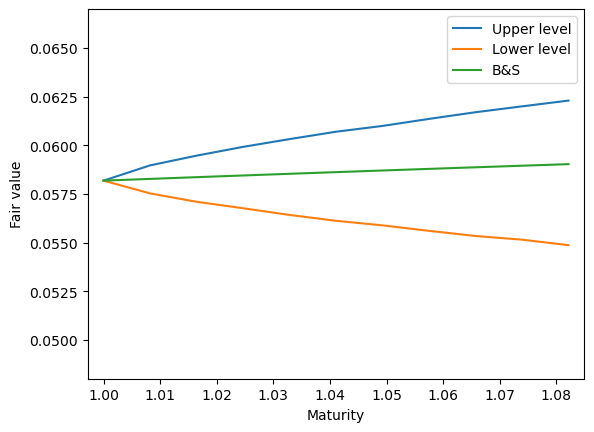}
	\caption{Black-Scholes option price with upper and lower bounds obtained by imposing the martingale constraint in the perturbation of the reference model.\ For the optimization, a neural network with width $20$ and depth $4$ has been used.\ Batch size: $2^{19}$ points.}
	\label{fig.option.bounds}
\end{figure}

\section{Proofs} \label{sec.proof.main}
In order to prove Theorem \ref{cor.main.itheta.single} and Theorem \ref{thm.main.itheta}, we need the following two auxiliary results.\ The first one is a version of a  standard result, cf.\ \cite{BartlEckstKuppRandWalks2021} and \cite{nendel2021wasspert}, specific to our setting. For the sake of a self-contained exposition, we provide a short proof.

\begin{lemma}\label{lem.restriction}
	Let $f\in \Lip_p$. Then, there exist $h_0>0$ and $a\geq 0$ (depending only on $f$ and the penalty function $\phi$) such that
	\begin{equation}\label{eq.restriction}
		\cI(h)f=\sup_{\cW_p(\mu,\nu)\leq ah} \left(\int_\X f(z)\, \nu(\d z)-\phi_h\big(\cW_p(\mu,\nu)\big)\right)\quad \text{for all }h\in (0,h_0].
	\end{equation}
\end{lemma}

\begin{proof}
	Since $f\in \Lip_p$, there exists a constant $L_f\geq 0$ such that
	\[
	|f(x_1)-f(x_2)|\leq L_f\big(1+\max\{\|x_1\|,\|x_2\|\}\big)^{p-1}\|x_1-x_2\|\quad \text{for all }x_1,x_2\in \X.
	\]
	By \eqref{eq: property of phi}, there exist constants $a \geq 0$ and $h_0>0$ such that
	\begin{equation}\label{eq.growthphi}
		\phi(u)>1+L_f(1+|\mu|_p+h_0 u)^{p-1} u\quad \text{for all } u\in (a,\infty).
	\end{equation}
	Let $h\in (0,h_0]$, $\nu\in \cP_p$ with
	\[
	\int_\X f(y)\,\mu(\d y) \le \cI(h)f\leq h+\int_\X f(z)\, \nu(\d z)-\phi_h\big(\cW_p(\mu,\nu)\big),
	\]
	and $\pi\in \cpl(\mu,\nu)$ be an optimal coupling between $\mu$ and $\nu$.
	Then, using H\"older's inequality and Minkowski's inequality,
	\begin{align*}
		\phi\left(\frac{\cW_p(\mu,\nu)}{h}\right)&\leq 1+\int_{\X\times \X}\frac{f(z)-f(y)}{h}\, \pi(\d y,\d z)\\
		&\leq 1+L_f \int_{\X\times \X} \big(1+\max\{\|y\|,\|z\|\}\big)^{p-1}\frac{\|z-y\|}{h}\, \pi(\d y,\d z)\\
		&\leq 1+L_f\bigg(\int_{\X\times \X} \big(1+\|y\|+\|z-y\|\big)^p\, \pi(\d y,\d z)\bigg)^{1/q}\frac{\cW_p(\mu,\nu)}{h}\\
		&\leq 1+L_f\bigg(1+|\mu|_p+h_0\frac{\cW_p(\mu,\nu)}{h}\bigg)^{p-1} \frac{\cW_p(\mu,\nu)}{h}.
	\end{align*}
	By \eqref{eq.growthphi}, it follows that $\frac{\cW_p(\mu,\nu)}{h}\leq a$.
\end{proof}

\begin{lemma}\label{lem.approx}
	Let $\Theta$ be dense in $L_p(\mu;\X)$ with $0\in \Theta$. Then,
	\[
	\cI_\Theta(h)f=\cI_{L_p(\mu;\X)}(h)f\quad\text{for all }f\in \Lip_p \text{ and }h>0.
	\]
\end{lemma}

\begin{proof}
	Let $h>0$ and $f \in \Lip_p$. 
	Clearly, $\cI_\Theta (h)f \le \cI_{L_p(\mu; \X)} (h) f$ since $\Theta \subset L_p(\mu; \X)$.
	To prove equality, fix $\eps > 0$, and let $\theta_0 \in L_p(\mu;\X)$ with
	\[
	\cI_{L_p(\mu;\X)}(h)f\leq  \int_\X f(z)\,\mu_{\theta_0}(\d z) - \phi_h(\| \theta_0 \|_{L_p(\mu;\X)})+\frac\eps3.
	\]
	Then, by dominated convergence, there exists some $\lambda\in (0,1)$ such that
	\[
	\int_\X f(z)\,\mu_{\theta_0}(\d z)\leq  \int_\X f(z)\,\mu_{\lambda \theta_0}(\d z) +\frac{\eps}3.
	\]
	Again, by dominated convergence and since $\Theta$ is dense in $L_p(\mu;\X)$ with $0\in \Theta$, there exists some $\theta\in \Theta$ with $\|\theta\|_{L_p(\mu;\X)}\leq \|\theta_0\|_{L_p(\mu;\X)}$ and
	\[
	\int_\X f(z)\,\mu_{\lambda \theta_0}(\d z)\leq \int_\X f(z)\,\mu_{\theta}(\d z)+\frac\eps3.
	\]
	Therefore, since $\phi$ is nondecreasing,
	\[
	\cI_{L_p(\mu;\X)}(h)f\leq  \int_\X f(z)\,\mu_{\theta}(\d z) - \phi_h(\| \theta \|_{L_p(\mu;\X)})+\eps\leq \cI_\Theta (h)f+\eps.
	\]
	Letting $\eps\downarrow 0$, we find that $\cI_\Theta(h)f=\cI_{L_p(\mu;\X)}(h)f$.
\end{proof}

\begin{proof}[Proof of Theorem \ref{cor.main.itheta.single}]
	The proof is divided in two steps. In a first step, we show that
	\begin{equation}\label{eq.proof.main.todo.1}
		\limsup_{h\downarrow 0} \frac{\cI(h)f-\mu f}h\leq \phi^*\Big( \big\| |f|_\Lip\big\|_{L_q(\mu)}\Big),
	\end{equation}
	and, in a second step, we prove
	\begin{equation}\label{eq.proof.main.todo.2}
		\liminf_{h\downarrow 0} \frac{\cI_\Theta(h)f-\mu f}h\geq \phi^*\Big( \big\| |f|_\Lip\big\|_{L_q(\mu)}\Big),
	\end{equation}
	whenever $\{b\theta\, |\, b\geq 0\}\subset \Theta$ for $\theta$ given by \eqref{eq.opt.theta}.
	In order to prove \eqref{eq.proof.main.todo.1}, let $h_0>0$ and $a\geq 0$ as in Lemma \ref{lem.restriction}. For $\delta>0$, let $|f|_{\Lip,\delta}\colon \X\to \R$ be given by
	\[
	|f|_{\Lip,\delta}(x):=\sup_{\gamma\in (0,\delta)}\sup_{\|u\|\in\{ 0,1\}}\frac{f(x+\gamma u)-f(x)}{\gamma}\quad\text{for all }x\in \X.
	\]
	Then, $|f|_{\Lip,\delta}$ is lower semicontinuous and thus measurable. Moreover, since $f\in \Lip_p$, for all $\delta>0$,
	\[
	|f|_{\Lip,\delta}(x)\leq \big(1+\delta+\|x\|\big)^{p-1}\quad\text{for all }x\in \X,
	\]
	which shows that $|f|_{\Lip,\delta}\in L^q(\mu)$.\ Let $\eps>0$.\ By dominated convergence, there exists some $\delta>0$, such that
	\begin{equation}\label{eq.proof.main.delta}
		\int_\X\big(|f|_{\Lip,\delta}(y)\big)^q\, \mu(\d y)\leq  \int_\X\big(|f|_{\Lip}(y)\big)^q\, \mu(\d y)+\eps.
	\end{equation}
	Let $h\in (0,h_0]$, $\nu\in \cP_p$ with $\cW_p(\mu,\nu)\leq ah$ and
	$$ \cI(h)f \leq \eps h+\int_\X f(z)\, \nu(\d z)-\phi_h\big(\cW_p(\mu,\nu)\big),$$
	and $\pi\in \cpl(\mu,\nu)$ be an optimal coupling between $\mu$ and $\nu$. Then, for arbitrary $\delta>0$,
	\begin{align}
		\notag \frac{\cI(h)f - \mu f}{h} & \leq \eps+\frac{1}{h} \int_{\X \times \X} f(z) - f(y)\,\pi (\d y, \d z) -  \phi\left(\frac{\cW_p(\mu, \nu)}{h}\right) \\
		\notag& = \eps+\frac{1}{h} \int_{\|z - y\| < \delta} f(z) - f(y)\,\pi (\d y, \d z) - \phi\left(\frac{\cW_p(\mu, \nu)}{h}\right) \\
		\label{eq.proof.main.estimate.1}& \quad + \frac{1}{h}\int_{\|z - y\| \ge \delta} f(z) - f(y)\,\pi(\d y, \d z).
	\end{align}
	We start by estimating the last term on the right-hand side in \eqref{eq.proof.main.estimate.1}.\ Since $f\in \Lip_p$,
	\begin{align*}
		\int_{\|z - y\| \ge \delta} f(z) - & f(y)\,\pi(\d y, \d z) \le L_f \int_{\|z - y\| \ge \delta} \big(1+\max\{\|y\|,\|z\|\}\big)^{p-1}\|z-y\|\, \pi(\d y,\d z) \\
		&\leq 2^{p-1}L_f\bigg(	\int_{\|z - y\| \ge \delta} \big(1+\|y\|\big)^{p-1}\|z-y\|\, \pi(\d y,\d z)+\cW_p(\mu,\nu)^p\bigg)\\
		& \leq 2^{p-1}L_f\cW_p(\mu,\nu)\bigg(	\int_{\|z - y\| \ge \delta} \big(1+\|y\|\big)^p\, \pi(\d y,\d z)+\cW_p(\mu,\nu)^{p-1}\bigg),
	\end{align*}
	where, in the last step, we used H\"older's inequality together with the fact that $u^{1/q}\leq u$ for all $u\in [1,\infty)$.\ By the monotone convergence theorem, there exists some some $\lambda\in (0,1)$ such that
	\[
	\int_{\X} \big(1+\|y\|\big)^p-\big(1+\|y\|\big)^{\lambda p}\, \mu(\d y)\leq \frac\eps2.
	\]
	Using Hölder's inequality with $\frac1\lambda$ and $\frac1{1-\lambda}$, Minkowski's inequality, and Lemma \ref{lem.restriction}, we find that
	\begin{align*}
		\int_{\|z - y\| \ge \delta} \big(1+\|y\|\big)^p\, \pi(\d y, \d z) &\le  \int_{\|z - y\| \ge \delta} \big(1+\|y\|\big)^{\lambda p}\, \pi(\d y,\d z) +\frac\eps2\\
		&\leq \frac{1}{\delta^{(1-\lambda)p}}\int_{\X\times \X} \big(1+\|y\|\big)^{\lambda p} \|z-y\|^{(1-\lambda)p}\, \pi(\d y,\d z)+\frac\eps2 \\
		&\leq \frac{1}{\delta^{(1-\lambda)p}}(1+|\mu|_p)^{\lambda p}\cW_p(\mu,\nu)^{(1-\lambda)p}+\frac\eps2\leq \eps
	\end{align*}
	for $h>0$ sufficiently small.\ Again, by Lemma \ref{lem.restriction}, we can estimate $\cW_p(\mu,\nu)^{p-1}\leq \eps$, for $h>0$ sufficiently small, and end up with
	\begin{equation}\label{eq.proof.main.firstterm}
		\int_{\|z - y\| \ge \delta} f(z) - f(y)\,\pi(\d y, \d z) \le 2^pL_f\eps \cW_p(\mu,\nu).
	\end{equation}
	In order to estimate the first term on the right-hand side in \eqref{eq.proof.main.estimate.1}, we use H\"older's inequality to observe that
	\begin{align*}
		\int_{\|z - y\| < \delta} {f(z) - f(y)}\,\pi(\d y, \d z) &\leq \int_{\X\times \X} |f|_{\Lip,\delta}(y)\|z-y\|\, \pi(\d y,\d z)\\
		&\leq \big\||f|_{\Lip,\delta}\big\|_{L_q(\mu)}\cW_p(\mu,\nu).
	\end{align*}
	Hence, a combination of \eqref{eq.proof.main.delta}, \eqref{eq.proof.main.estimate.1}, and \eqref{eq.proof.main.firstterm} yields that
	\begin{align*}
		\frac{\cI(h)f - \mu f}{h}& \le \eps +  \phi^*\bigg(2^{p}L_f\eps+\big\||f|_{\Lip,\delta}\big\|_{L_q(\mu)}\bigg) \\
		& \leq \eps +  \phi^*\bigg(\big(1+2^{p}L_f\big)\eps+\big\||f|_{\Lip}\big\|_{L_q(\mu)}\bigg).
	\end{align*}
	The bound from above \eqref{eq.proof.main.todo.1} is then proved thanks to the arbitrariness of $\eps >0$ and the continuity of $\phi^*$.
	
	For the bound from below \eqref{eq.proof.main.todo.2}, consider the function $\theta\in L_p(\mu;\X)$, which, for $x\in \X$, is given by
	\[
	\theta(x) = v(x) \big(| f |_\Lip(x)  \big)^{q-1},
	\]
	where $v\colon \X\to \X$ is a measurable direction of steepest ascent for $f$.\ We observe that, by definition of $\theta$, $\int_\X \|\theta(y)\|^p\, \mu(\d y) = \big\| |f|_\Lip \big\|_{L_q(\mu)}^q$ and, for $b\geq0$,
	\[
	\frac{\cI_\Theta(h)f - \mu f}{h} \ge \int_\X \frac{f\big(y+bh\theta(y)\big) - f(y)}{h}\, \mu(\d y) -\varphi\big(b\|\theta\|_{L_p(\mu;\X)}\big),
	\]
	which, by Fatou's lemma, leads to
	\begin{align*}
		\liminf_{h \downarrow 0} \frac{\cI_\Theta(h)f - \mu f}{h} & \ge b \int_\X \big(|f|_\Lip(y)\big)^q\,\mu(\d y) - \phi\big(b\|\theta\|_{L_p(\mu;\X)}\big) \\
		& = b \big\||f|_\Lip\big\|^q_{L_q(\mu)} -\phi\Big(b \big\| |f|_\Lip \big\|_{L_q(\mu)}^{q - 1}\Big)\quad\text{for all }b\geq 0.
	\end{align*}
	We pass to the convex conjugate $\phi^*$ of $\phi$ by taking the supremum over all $b\geq 0$, and end up with \eqref{eq.proof.main.todo.2}.
	
\end{proof}

\begin{proof}[Proof of Theorem \ref{thm.mean.constraint}]
	Consider the auxiliary function
	\[
	\tilde{f}(x) = f(x) - \Big\langle \int_\X \nabla f(y)\,\mu(\d y), x - \int_\X y\,\mu(\d y) \Big\rangle\quad\text{for all }x\in \X,
	\]
	and note that $\nu f = \nu \tilde{f}$ for all $\nu\in \cP_p(\mu)$. In particular, $\cI^\text{Mean}(h) f=\cI^\text{Mean}(h) \tilde{f}$ for all $h>0$.
	Since the supremum in the definition of $\cI(h)$ runs over a larger set than the one in the definition of $\cI^\text{Mean}(h)$, we have $\cI^\text{Mean}(h) \tilde{f} \le \cI(h) \tilde{f}$ for all $h > 0$. 
	Define 
	\begin{equation} \label{eq.optimiz.direct.mean}
		\theta(x) := \nabla \tilde{f}(x)=\nabla f(x) - \int_\X \nabla f(y)\,\mu(\d y)\quad\text{for all }x\in \X,
	\end{equation}
	and let $\Theta:=\{b\theta\, |\, b\geq 0\}$.\
	By definition, $\int_\X\theta(y)\, \mu(\d y)=0$, so that $\cI_{\Theta}(h) f \leq  \cI^\text{Mean}(h) f$ and $\cI_{\Theta}(h) f =  \cI_\Theta(h)\tilde f$ for all $h>0$. This implies that
	\[
	0\leq \frac{\cI^\text{Mean}(h) f - \cI_\Theta(h) f}h \leq \frac{\cI(h) \tilde{f} - \cI_\Theta(h) \tilde{f}}h\quad\text{for all }h>0,
	\]
	and therefore, by Theorem \ref{thm.main.itheta},
	\[
	\lim_{h\downarrow0}\frac{\cI^\text{Mean}(h) f - \cI_\Theta(h) f}h =\lim_{h\downarrow0} \frac{\cI(h) \tilde{f} - \cI_\Theta(h) \tilde{f}}h=0.
	\]
	Again, by Theorem \ref{thm.main.itheta}, it follows that 
	\[
	\lim_{h\downarrow0}\frac{\cI^\text{Mean}(h) f-\mu f}h=\lim_{h\downarrow0}\frac{\cI(h) \tilde{f} - \mu \tilde{f}}h=\phi^*\Big(\big\|\nabla \tilde f\big\|_{L_p(\mu;\X)}\Big).
	\]
\end{proof}

Before turning our focus on the proof of Theorem \ref{thm.martingale.constraint}, we provide the following two modifications of Lemma \ref{lem.restriction} and Lemma \ref{lem.approx}.

\begin{lemma} \label{lem.restriction.2nd}
	Let $f \in C^2_p$.\ Then, there exist $h_0>0$ and constant $a \geq 0$ (depending only on $f$ and the penalty function $\phi$) such that
	\[
	\cI^\mart(h)f = \sup_{\cW_p^\mart (\mu,\nu)\leq \sqrt{a h}} \left( \int_\X f(z)\, \nu(\d z) - \phi_h\big(\cW_p^\mart(\mu, \nu)^2\big) \right)\quad \text{for all }h\in (0,h_0].
	\]
\end{lemma}

\begin{proof}
	The proof is similar to the one of Lemma \ref{lem.restriction}.\ Since $f \in C^2_p$, there exists a constant $C \geq 0$ such that
	\[
	\big\|\nabla^2 f(x)\big\| \le C(1 + \|x\|)^{p-2} \quad \text{for all } x \in \X.
	\]
	Moreover, by Equation \eqref{eq: property of phi_v2}, there exist constants $a \ge 0$ and $h_0 > 0$ such that
	\begin{equation} \label{eq.growthphi.squared}
		\phi(u^2) > 1 + \frac{C}{2}\Big(1 + |\mu|_p + \sqrt{h_0} u\Big)^{p-2} u^2 \quad \text{for all } u \in (\sqrt{a}, \infty).
	\end{equation}
	Let $h \in (0, h_0]$, $\nu \in \cP_p$ with
	\[
	\int_\X f(y)\,\mu(\d y) \le \cI(h)f \le h + \int_\X f(z)\,\nu(\d z) - \phi_h\big(\cW_p^\text{Mart}(\mu,\nu)^2\big),
	\]
	and $\pi \in \cpl(\mu, \nu)$ be an optimal martingale coupling between $\mu$ and $\nu$.\ Using Taylor's theorem,
	\begin{align*}
		\phi\left( \frac{\cW_p^\text{Mart}(\mu, \nu)^2}{h} \right) & \le 1 + \int_{\X \times \X} \frac{f(z) - f(y)}{h}\,\pi(\d y, \d z) \\
		& = 1 + \frac{1}{h} \bigg(\int_{\X \times \X} \big\langle \nabla f(y), z - y \big\rangle \, \pi(\d y, \d z) \\
		& \quad\; + \int_{\X \times \X} \int_0^1 (1-t) \big\langle \nabla^2 f\big(y + t(z-y)\big) (z - y), z - y \big\rangle \,\d t \, \pi(\d y, \d z) \bigg).
	\end{align*}
	Thanks to the martingale constraint, the first term in the Taylor expansion vanishes, so that,
	using Fubini's theorem, Hölder's inequality, and Minkowski's inequality, we obtain
	\begin{align*}
		\phi\Bigg( & \frac{\cW_p^\text{Mart}(\mu, \nu)^2}{h} \Bigg) \le 1 + \frac{1}{h} \int_0^1 (1-t) \int_{\X \times \X} \big|\nabla^2 f(y+t(z-y))\big|_\Max \|z - y\|^2\,\pi(\d y, \d z)\,\d t \\
		& \qquad \le 1 + \int_0^1 (1-t) \left( \int_{\X \times \X} \big|\nabla^2 f(y+t(z-y))\big|_\Max^\frac{p}{p-2}\,\pi(\d y, \d z) \right)^\frac{p-2}{p}\,\d t \frac{\cW_p^\text{Mart}(\mu, \nu)^2}{h} \\
		& \qquad \le 1 + \frac{C}{2} \left( \int_{\X \times \X} \big(1 + \|y\| + \|z - y\|\big)^p\,\pi(\d y,\d z) \right)^\frac{p-2}{p} \frac{\cW_p^\text{Mart}(\mu, \nu)^2}{h} \\
		& \qquad \le 1 + \frac{C}{2} \left( 1 + |\mu|_p + \sqrt{h_0}\frac{\cW_p^\text{Mart}(\mu, \nu)}{\sqrt{h}} \right)^{p - 2} \frac{\cW_p^\text{Mart}(\mu, \nu)^2}{h}.
	\end{align*}
	By Equation \eqref{eq.growthphi.squared}, it follows that $\frac{\cW_p^\text{Mart}(\mu, \nu)}{\sqrt h} \le \sqrt{a}$.
\end{proof}

\begin{lemma}\label{lem.approx2}
	Let $\Theta$ be dense in $L_p(\mu;\X)$ with $0\in \Theta$. Then,
	\[
	\cI^\mart_\Theta(h)f=\cI^\mart_{L_p(\mu;\X)}(h)f\quad\text{for all }f\in \Lip_p \text{ and }h>0.
	\]
\end{lemma}

\begin{proof}
	We proceed exactly as in the proof of Lemma \ref{lem.approx}. Let $h>0$ and $f \in C_p^2$. 
	Again, $\cI^\mart_\Theta (h)f \le \cI^\mart_{L_p(\mu; \X)} (h) f$, since $\Theta \subset L_p(\mu; \X)$.
	Let $\eps > 0$ and $\theta_0 \in L_p(\mu;\X)$ with
	\[
	\cI^\mart_{L_p(\mu;\X)}(h)f\leq  \int_\X f(z)\,\mu^\mart_{\theta_0}(\d z) - \phi_h(\| \theta_0 \|^2_{L_p(\mu;\X)})+\frac\eps3.
	\]
	By dominated convergence, we first find $\lambda\in (0,1)$ with
	\[
	\int_\X f(z)\,\mu^\mart_{\theta_0}(\d z)\leq  \int_\X f(z)\,\mu^\mart_{\lambda \theta_0}(\d z)
	+\frac{\eps}3,
	\]
	and then $\theta\in \Theta$ with $\|\theta\|_{L_p(\mu;\X)}\leq \|\theta_0\|_{L_p(\mu;\X)}$ and
	\[
	\int_\X f(z)\,\mu^\mart_{\lambda \theta_0}(\d z)\leq \int_\X f(z)\,\mu^\mart_{\theta}(\d z)+\frac\eps3,
	\]
	where we used the fact that $\Theta$ is dense in $L_p(\mu;\X)$ and $0\in \Theta$.  
	Since $\phi$ is nondecreasing,
	\[
	\cI^\mart_{L_p(\mu;\X)}(h)f\leq  \int_\X f(z)\,\mu^\mart_{\theta}(\d z) - \phi_h(\| \theta \|_{L_p(\mu;\X)}^2)+\eps\leq \cI^\mart_\Theta (h)f+\eps.
	\]
	Letting $\eps\downarrow 0$ the claim follows.
\end{proof}

\begin{proof}[Proof of Theorem \ref{thm.martingale.constraint}]
	
	We use a similar approach as in the proof of Theorem \ref{cor.main.itheta.single}, providing a bound from above for $\cI^\mart(h)f-\mu f$ and a bound from below for $\cI^\mart_\Theta(h)f-\mu f$ in the limit $h\downarrow 0$. We start with the bound from above writing a Taylor expansion for $f$ as we did in the proof of Lemma \ref{lem.restriction.2nd}. To that end, let $a\geq 0$ and $h_0>0$ as in Lemma \ref{lem.restriction.2nd}. For all $h\in (0,h_0]$, let $\nu_h\in \cM_p(\mu)$ with
	$\cW^\mart_p(\mu,\nu_h)\leq  \sqrt{a h}$ and
	\[
	\cI^\text{Mart}(h)f\leq h^2+ \int_\X f(z)\, \nu_h(\d z)-\phi_h\big(\cW^\mart(\mu,\nu_h)^2\big),
	\]
	and $\pi_h\in \mart(\mu,\nu)$ be a martingale coupling between $\mu$ and $\nu_h$ that attains the infimum in definition of the $p$-martingale Wasserstein distance between $\mu$ and $\nu_h$, cf.\ \cite[Theorem 1.7]{Beiglbck2016}.\ Then, for all $h\in (0,h_0]$,
	\begin{align*}
		\frac{\cI^\text{Mart}(h)f - \mu f}{h} & \leq h+\frac{1}{h} \int_{\X \times \X} f(z) - f(y) \, \pi_h(\d y, \d z) - \phi\bigg(\frac{\cW_p^\text{Mart}(\mu, \nu_h)^2}{h}\bigg) \\
		& = h + \frac1h \int_{\X \times \X} \int_0^1 (1-t) \big\langle \nabla^2 f\big(y + t(z-y)\big) (z - y), z - y \big\rangle \,\d t \, \pi_h(\d y, \d z) \\
		& \quad- \phi\bigg(\frac{\cW_p^\text{Mart}(\mu, \nu_h)^2}h\bigg),
	\end{align*}
	where, thanks to the martingale constraint, the first order term in the Taylor expansion vanishes as in the proof of Lemma \ref{lem.restriction.2nd}.
	
	In the second order term, we renormalize the vector $z-y$ multiplying and dividing by $\|z - y\|$.
	Then, using Fubini's theorem and Hölder's inequality, we obtain, for all $h\in (0,h_0]$,
	\begin{align*}
		\frac{1}{h} \int_{\X \times \X} &\int_0^1  (1-t) \big\langle \nabla^2 f\big(y + t(z-y)\big) (z - y), z - y \big\rangle \,\d t \, \pi_h(\d y, \d z) \\
		& \le \frac{1}{h} \int_0^1 \int_{\X \times \X} (1-t) \big|\nabla^2 f(y + t(z - y))\big|_\Max \| z - y \|^2 \,\pi_h(\d y, \d z) \,\d t \\
		& \le \int_0^1 (1-t) \left( \int_{\X \times \X} \big|\nabla^2 f(y + t(z - y))\big|_\Max^{\frac{p}{p-2}}\,\pi_h(\d y, \d z) \right)^{\frac{p-2}p}\,\d t \frac{\cW_p^\text{Mart}(\mu, \nu_h)^2}{h}.
	\end{align*}
	Substituting the last inequality in the previous computations and passing to the convex conjugate of $\phi$, we find that
	\[
	\frac{\cI^\text{Mart}(h)f - \mu f}{h} \le h+ \phi^*\left( \int_0^1 (1-t)\left( \int_{\X \times \X} \big|\nabla^2 f(y + t(z - y))\big|_\Max^{\frac{p}{p-2}} \, \pi_h(\d y, \d z) \right)^{\frac{p-2}{p}} \,\d t \right)
	\]
	for all $h\in (0,h_0]$. By assumption, $f$ is an element of $C_p^2$, so that $|\nabla^2f|_{\Max}\colon \X\to \R$ is continuous, and there exists a constant $C\geq 0$ such that $$\big|\nabla^2 f(x)\big|_\Max^{\frac{p}{p-2}} \le \|\nabla^2 f (x)\|^{\frac{p}{p-2}} \le C(1 + \|x\|)^p\quad\text{for all }x\in \X.$$ Moreover,
	\[
	\bigg(\int_{\X\times \X}\|y-z\|^p\, \pi_h(\d \mu,\d \nu_h)\bigg)^{1/p}=\cW_p^\mart(\mu,\nu_h)\leq \sqrt{ah}\quad \text{for all }h\in (0,h_0].
	\]
	Therefore, $\pi_h$ concentrates on the diagonal as $h\downarrow 0$, i.e.,
	\[
	\int_{\X\times \X} g(y,z)\,\pi_h(\d y,\d z)\to  \int_{\X\times \X} g(y,y)\, \mu(\d y) \quad\text{as }h\downarrow 0
	\]
	for all continuous functions $g\colon \X\times \X\to \R$ with $|g(x_1,x_2)|\leq M(1+\|x_1\|+\|x_2\|)^p$ for all $x_1,x_2\in \X$ and some constant $M\geq 0$, cf.\ \cite[Theorem 6.9]{villani2008optimal}, and the continuity of $\phi^*$ implies that
	\[
	\limsup_{h\downarrow 0}\frac{\cI^\text{Mart}(h)f - \mu f}{h} \le \phi^*\left(\frac{1}{2} \left(\int_{\X \times \X} \big|\nabla^2 f(y)\big|_\Max^{\frac{p}{p-2}} \, \mu(\d y) \right)^{\frac{p-2}{p}} \right).
	\]
	For the bound from below, let $\eps>0$ and $v \colon \X \to \S_0$ satisfying \eqref{eq.measurable curv} in Remark \ref{rem.mart.meas.select}, and define
	\begin{equation} \label{eq.optimiz.direct.mart}
		\theta(x) := v(x) \big(\big|\nabla^2 f(x)\big|_\Max\big)^\frac{1}{p-2}.
	\end{equation}
	Since $f \in C^2_p$, it follows that $\theta \in L_p(\mu; \X)$, so that, for every $b \geq 0$,
	\[
	\frac{\cI^\text{Mart}_\Theta (h)f - \mu f}{h} \ge \int_{\X}\int_\R \frac{f\big(y+s\sqrt{bh}\theta(y)\big) - f(y)}{h}\,B_{\sym}(\d s)\, \mu(\d y)  - \phi\Big( b\|\theta\|^2_{L_p(\mu;\X)}\Big). 
	\]
	Using Fatou's lemma, this leads to
	\begin{align*}
		\liminf_{h\downarrow 0} & \frac{\cI^\text{Mart}_\Theta (h)f - \mu f}{h}\geq \int_{\X} \frac{b}2|\nabla^2f(y)|_\Max^{\frac{p}{p-2}}\,\mu(\d y)  - \phi\Big( b\|\theta\|^2_{L_p(\mu;\X)}\Big)-\eps\\
		& \qquad \qquad = \frac{b}2\int_{\X}|\nabla^2f(y)|_\Max^{\frac{p}{p-2}}\,\mu(\d y) - \phi\left( b\bigg(\int_{\X}|\nabla^2f(y)|_\Max^{\frac{p}{p-2}}\,\mu(\d y)\bigg)^{\frac2p}\right)-\eps
	\end{align*}
	for all $b\geq0$. Taking the supremum over all $b\geq 0$, we can pass to the conjugate $\phi^*$ of $\phi$, and obtain that
	\[
	\liminf_{h\downarrow 0}\frac{\cI^\text{Mart}_\Theta (h)f - \mu f}{h}\geq \phi^*\left(\frac12 \bigg(\int_{\X}|\nabla^2f(y)|_\Max^{\frac{p}{p-2}}\,\mu(\d y)\bigg)^{\frac{p-2}p}\right)-\eps.
	\]
	Taking the limit $\eps\downarrow 0$ and invoking Lemma \ref{lem.approx2} completes the proof.
\end{proof}

\medskip
\noindent\textbf{Competing interests}\smallskip\

\noindent The authors have no relevant financial or non-financial interests to disclose.

\appendix
	
\section{A simple but useful approximation result}\label{app.A}

\begin{proposition}\label{prop.append.approx}
	Let $(\Omega, \mathcal F,\mu)$ be a $\sigma$-finite measure space, $(S_n)_{n\in \N}\subset \mathcal F$ with $\mu(S_n)<\infty$ for all $n\in \N$ and $\Omega=\bigcup_{n\in \N}S_n$, $X$ a Banach space, endowed with the Borel $\sigma$-algebra, and $\mathcal A$ an $\cap$-stable generator
	of $\cF$ with $\Omega\in \cA$. Let $p\in [1,\infty)$ and $L\subset L_p(\mu)$ with $\id_{A\cap S_n}\in \overline{\spn L}$ (closure of $\spn L$ in $L_p(\mu)$) for all $A\in \cA$ and $n\in \N$. Then, $$\overline{\Xspn L}=L_p(\mu;X),$$
	where $\Xspn L:={\rm spn}\{x\cdot f\,|\, x\in X,\, f\in L\}$.
\end{proposition}

\begin{proof}
	Let $\mathcal D: = \{B\in \mathcal F\,|\, \id_{B\cap A_n} \in \overline{\spn L}\text{ for all }n\in \N\}$. By assumption,
	$A\subset \mathcal  D$ and, in particular, $\Omega\in \mathcal D$. Let $B\in \mathcal D$ and $n\in \N$. Then, there exist sequences $(f_k)_{k\in \N}\subset \spn L$ and $(g_k)_{k\in \N}\subset \spn L$ with
	$$\lim_{k\to \infty}\|f_k-\id_{S_n}\|_{L_p(\mu)}=0\quad \text{and}\quad \lim_{k\to \infty}\|g_k-\id_{B\cap S_n}\|_{L_p(\mu)}=0.$$
	Then, $f_k-g_k\in \spn L$ for all $k\in \N$, and
	$$
	\|f_k-g_k-\id_{B^c\cap S_n}\|_{L_p(\mu)}\leq \|f_k-\id_{S_n}\|_{L_p(\mu)}+\|g_k-\id_{B\cap S_n}\|_{L_p(\mu)}\to 0\quad \text{as }k\to \infty.
	$$
	Now, let $(A_k)_{k\in \N}\subset \cD$ pairwise
	disjoint, $\eps> 0$, $n\in \N$, $A_k^n:=A_k\cap S_n$ for all $k\in \N$, and $A^n:=\dot\bigcup_{k\in \N}A_k^n$. By $\sigma$-additivity of $\mu$, there exists some $k \in \N$ with $$\bigg\|\id_{A^n}-\sum_{i=1}^k\id_{A_i^n}\bigg\|_{L_p(\mu)}<2^{-k}\eps.$$
	On the other hand, there exists $f_1,\ldots,f_n\in \spn L$ with $\|f_i-\id_{A_i^n}\|_{L_p(\mu)}<2^{-i}\eps$ for all $i\in \{1,\ldots,k\}$. Therefore,
	\begin{align*}
		\bigg\|\id_{A^n}-\sum_{i=1}^k f_i \bigg\|_{L_p(\mu)}&\leq \bigg\|\id_{A^n}-\sum_{i=1}^k\id_{A_i^n}\bigg\|_{L_p(\mu)}+\sum_{i=1}^k\big\|\id_{A_i^n}-f_i\big\|_{L_p(\mu)}<\eps.
	\end{align*}
	This shows that $\dot\bigcup_{k \in \N}A_k \in \cD$. By Dynkin's lemma, $\cD = \cF$. Now the statement
	follows from the fact that $L_p(\mu;X)$ is the closure of ${\rm spn}\{ x\cdot \id_{B\cap S_n}\,|\, x\in X,\, B\in \cF, \, n\in \N\}$, which is a consequence of the dominated convergence theorem, cf.\ \cite[Proposition 1.2.5]{MR3617205}, and \cite[Lemma 1.2.19(1)]{MR3617205}.
\end{proof}

As a consequence of the previous proposition, we obtain the following approximation results.
\begin{remark}\ \label{rem.append.approx.ex}
	Let $(\Omega, \cF,\mu)$ be a $\sigma$-finite measure space, $X$ a Banach space, endowed with the Borel $\sigma$-algebra, and $p\in [1,\infty)$.
	\begin{enumerate}[a)]
		\item As a consequence of Proposition \ref{prop.append.approx}, we obtain the following well-known approximation result, see, e.g., \cite[Lemma A.1]{MR4319244} for a proof in the special case $X=\R$: Let $\Omega$ be a metric space with the Borel $\sigma$-algebra $\mathcal F$ and a finite measure $\mu$.\ Then, $\Xspn \Lip_b(\Omega)$ is dense in $L_p(\mu;X)$, where $\Lip_b(\Omega)$ denotes the space of all bounded Lipschitz continuous functions $\Omega\to \R$. In particular, if $\Omega$ or, alternatively, $X$ is assumed to be separable, the space $\Lip_b(\Omega;X)$ of all bounded Lipschitz continuous functions $\Omega\to X$ is a dense subset of $L_p(\mu;X)$.\ In order to deduce this approximation result from Proposition \ref{prop.append.approx}, recall that, for every open set $U\subset \Omega$, $\id_U$ can be obtained as the pointwise monotone limit of bounded Lipschitz continuous functions using the so-called $\inf$-convolution
		$$
		f_\delta(\omega):=\inf_{\omega_0\in \Omega}\bigg(\id_U(\omega_0)+\frac{d(\omega,\omega_0)}{\delta}\bigg)\quad \text{for }\omega\in \Omega \text{ and }\delta>0.
		$$
		\item If $\Omega=H$ is a separable Hilbert space, $\mathcal F$ is the Borel $\sigma$-algebra, and every ball in $H$ has finite measure under $\mu$, then $\Xspn C_{b,b}^\infty(H)$ is dense in $L_p(\mu;X)$, where $C_{b,b}^\infty(H)$ denotes the space of all bounded infinitely smooth functions $H\to \R$ with bounded support.\ In particular, the space $C_{b,b}^\infty(H;X)$ of all bounded infinitely smooth functions $H\to X$ with bounded support is dense in $L_p(\mu;X)$.
		This approximation result follows directly from Proposition \ref{prop.append.approx} together with the fact that the indicator function of every open ball with radius $r>0$ and barycenter $x_0\in H$ is the pointwise monotone limit of the functions
		\[
		f_\delta(x):=\begin{cases}
			\exp\Big({-\tfrac{\delta\|x-x_0\|^2}{r^2-\|x-x_0\|^2}}\Big),& \|x-x_0\|<r,\\
			0,&\|x-x_0\|\geq r,
		\end{cases} \quad \text{for }x\in H\text{ and }\delta>0,
		\]
		and finite intersections of such balls are monotone limits of finite products of such functions. 
		\item \label{rem.append.approx.ex.options}Let $\Omega=\R$, $\mathcal F$ be the Borel $\sigma$-algebra, and $\mu$ be a finite measure. Then, for $$L_0:=\big\{1_{[K,\infty)}\, \big|\, K\in \R\big\},$$ the space $\Xspn L_0$ is dense in $L_p(\mu;X)$. If, additionally, $\int_\R |y|^p\, \mu(\d y)<\infty$ and $$ L_1:=\Big\{\big[x\mapsto (x-K)^+\big]\, \Big|\, K\in \R\Big\},$$
		then also $\Xspn L_1$ is dense in $L_p(\mu;X)$.
		\item Let $\mu$ be a finite measure and $\mathcal A$ an $\cap$-stable generator of $\mathcal F$. Then, $\Xspn \{ \id_A\, |\, A\in \mathcal A\}$ is dense in $L_p(\mu;X)$.\ Note that this approximation holds for an arbitrary $\cap$-stable generator $\mathcal A$ without any additional assumptions.\ As a consequence, for $X=\R$, one obtains the well-known Carath\'eodory approximation: if $\mu$ is a finite measure and $\mathcal A$ is an algebra with $\sigma(\mathcal A)=\mathcal F$, then, for all $B\in \mathcal F$ and all $\eps>0$, there exists some $A\in \mathcal A$ with
		\[
		\mu(A\Delta B)=\|\id_A-\id_B\|_{L_1(\mu)}<\eps.
		\]
		Indeed, let $f\in \spn \{ \id_A\, |\, A\in \mathcal A\}$ with $\|\id_B-f\|_{L_1(\mu)}<\eps/2$. Then, since $\mathcal A$ is an algebra, $A:=\{\omega\in \Omega\, |\, f(\omega)> 1-\eps/2\}$ is an element of $\mathcal A$, and
		\[
		\mu(A\Delta B)=\mu(A\setminus B)+\mu(B\setminus A)<\eps.
		\]
	\end{enumerate}
\end{remark}

\bibliographystyle{abbrv}

\end{document}